\documentclass[acmsmall, nonacm]{acmart}

\AtBeginDocument{%
  \providecommand\BibTeX{{%
    \normalfont B\kern-0.5em{\scshape i\kern-0.25em b}\kern-0.8em\TeX}}}

\setcopyright{acmcopyright}
\copyrightyear{0000}
\acmYear{0000}
\acmDOI{00.0000/0000000.0000000}

\acmJournal{JACM}
\acmVolume{00}
\acmNumber{0}
\acmArticle{000}
\acmMonth{12}

\usepackage{amsthm}
\DeclareMathOperator*{\var}{Var}	
\DeclareMathOperator*{\bias}{Bias}	
\usepackage{dsfont} 
\usepackage{paralist} 
\usepackage{url} 
\usepackage{pdflscape} 
\usepackage{ragged2e} 

\newtheorem{theorem}{Theorem}

\newtheorem{corollary}[theorem]{Corollary}

\newtheorem{remark}{Remark} 
\newtheorem{assumption}{Assumption}
\newtheorem*{problemdefinition}{Problem Definition} 

\usepackage[ruled,vlined]{algorithm2e}

\usepackage{enumerate} 

\usepackage{graphicx}
\usepackage[outdir=./]{epstopdf}
\usepackage{subcaption}

\usepackage{tabulary}
\usepackage{booktabs}

\newcommand{\degdist}{p} 
\newcommand{\neighbordegdist}{q} 
\newcommand{\empiricaldegdist}{{\hat{p}}} 
\newcommand{\empiricalneighbordegdist}{{\hat{q}}} 
\newcommand{\tailscope}{\hat{p}_{\text{tail-scope}}}
\newcommand{\samplesize}{n} 
\newcommand{\mindegree}{k_{\text{min}}} 
\newcommand{\vanillaMLE}{\hat{\alpha}_{\text{vl}}} 
\newcommand{\vanillaDiscreteMLE}{\hat{\alpha}_{\text{vl-d}}} 
\newcommand{\vanillaMLEexponential}{\hat{\lambda}_{\text{vl}}} 
\newcommand{\fpMLE}{\hat{\alpha}_{\text{fp}}} 
\newcommand{\fpDiscreteMLE}{\hat{\alpha}_{\text{fp-d}}} 
\newcommand{\fpMLEexponential}{\hat{\lambda}_{\text{fp}}} 
\newcommand{\loglikelihoodX}{\mathcal{L}_{\text{vl}}} 
\newcommand{\loglikelihoodY}{\mathcal{L}_{\text{fp}}} 
\newcommand{\CRLBvanilla}{{\text{CRLB}}_{\text{vl}}} 
\newcommand{\CRLBfp}{{\text{CRLB}}_{\text{fp}}} 
\newcommand{\ntail}{n_{\mathrm{tail}}} 
\newcommand{\CCDF}{P} 

\fancypagestyle{mylandscape}{
	\fancyhf{} 
	\fancyfoot{
		\makebox[\textwidth][r]{
			\rlap{\hspace{.75cm}
				\smash{
					\raisebox{4.87in}{
						\rotatebox{90}{\thepage}}}}}}
}

\begin{document}

\title{Maximum Likelihood Estimation of  Power-law Degree Distributions via Friendship Paradox based Sampling}

\author{Buddhika Nettasinghe}
\email{dwn26@cornell.edu}
\author{Vikram Krishnamurthy}
\email{vikramk@cornell.edu}
\affiliation{%
	\institution{School of Electrical and Computer Engineering, Cornell University}
	\streetaddress{Frank H.T. Rhodes Hall}
	\city{Ithaca}
	\state{New York}
	\postcode{14850}
}

\renewcommand{\shortauthors}{Nettasinghe and Krishnamurthy, et al.}

\begin{abstract}
This paper considers the problem of estimating a power-law degree distribution of an undirected network using sampled data. Although power-law degree distributions are ubiquitous in nature, the widely used parametric methods for estimating them~(e.g.~linear regression on double-logarithmic axes, maximum likelihood estimation with uniformly sampled nodes) suffer from the large variance introduced by the lack of data-points from the tail portion of the power-law degree distribution. As a solution, we present a novel maximum likelihood estimation approach that exploits the \textit{friendship paradox} to sample more efficiently from the tail of the degree distribution. We analytically show that the proposed method results in a smaller bias, variance and a  Cram\`{e}r-Rao lower bound compared to the vanilla maximum-likelihood estimate obtained with uniformly sampled nodes~(which is the most commonly used method in literature). Detailed numerical and empirical results are presented to illustrate the performance of the proposed method under different conditions and how it compares with alternative methods. We also show that the proposed method and its desirable properties (i.e.~smaller bias, variance and Cram\`{e}r-Rao lower bound compared to vanilla method based on uniform samples) extend to parametric degree distributions other than the power-law such as exponential degree distributions as well. All the numerical and empirical results are reproducible and the code is publicly available on Github.
\end{abstract}



\maketitle

\section{\label{sec:introduction}Introduction}

\begin{figure*}[!tbh]
	\includegraphics[width=1.0\columnwidth, trim={0.4cm 0.2cm 0.15cm 0.0cm},clip]{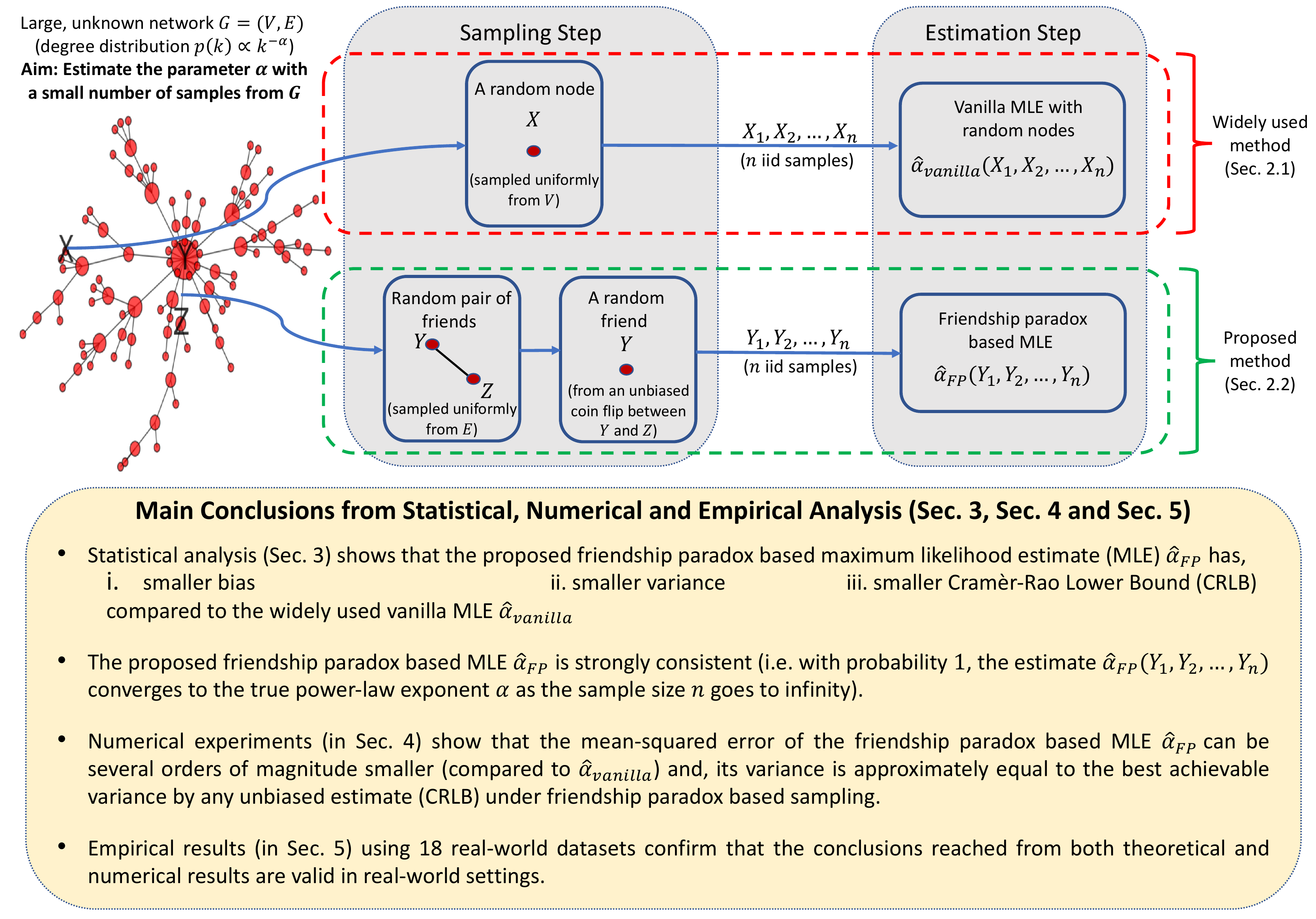}
	\caption{Summary of the key idea of this paper (maximum likelihood estimation with friendship paradox based sampling for estimating power-law degree distributions), how it compares with the widely used existing method, the conclusions drawn from its statistical and numerical analysis and how these ideas are organized in different sections of the paper. The widely used vanilla method samples nodes uniformly and therefore, fails to efficiently capture the characteristic heavy tail of the power-law degree distribution due to the lack of samples from that regime. In comparison, the proposed method exploits the parametric form of the degree distribution and the friendship paradox based sampling~(detailed in Sec.~\ref{subsec:fp}) to obtain nodes with larger degrees with higher probabilities (indicated by node sizes in the figure) which then leads to a sample containing more high degree nodes that efficiently captures the characteristic heavy tail of the power-law degree distribution. Our results are backed by statistical analysis and detailed numerical simulation experiments and also generalize to parametric forms other than power-law.  }
	\label{fig:blockdiag}
\end{figure*}

Many real-world networks such as social networks~\cite{aiello2000random, ebel2002scale, leskovec2007graph}, internet~\cite{faloutsos1999power}, world-wide web~\cite{huberman1999internet} power-grids~\cite{chassin2005evaluating}, and biological networks~\cite{jeong2001lethality, eguiluz2005scale} have power-law degree distributions~i.e.~the probability~$p(k)$ that a uniformly sampled node has $k$ neighbors is proportional to $k^{-\alpha}$ for a fixed value of the power-law exponent $\alpha > 0$. A key reason for this ubiquity is that power-law distributions arise naturally from simple and intuitive generative processes such as preferential attachment~\cite{barabasi1999emergence,  mitzenmacher2004brief, newman2005power, gabaix2009power}. However, despite being ubiquitous in nature, the full structure of such networks is often not precisely known due to practical constraints such as large size, restrictions on accessing network data, etc.~\cite{murai2013set}.
Hence, estimating the power-law exponent~$\alpha$ from small~(compared to size of the network) and noisy samples of data is a key step in the study of networks and related topics such as spreading processes on networks~\cite{pastor2001epidemic, pastor2015epidemic, gomez2012inferring} and network stability~\cite{cohen2000resilience}. 
This problem is defined formally as follows:

\begin{problemdefinition}[Estimating the exponent of a power-law degree distribution]
	Consider an undirected network ${G = (V,E)}$ with a power-law degree distribution
	\begin{align}
	p(k)\propto k^{-\alpha},  \quad k \geq k_{min} \label{eq:pk_proportional_k}
	\end{align} 
	where the power-law exponent ${\alpha>0}$ is unknown. Assume ${k_{min}>0}$ is the known minimum degree. Estimate the power-law exponent $\alpha$ using the degrees of~$\samplesize$~nodes independently sampled from the network~$G$.
\end{problemdefinition}


\noindent
{\bf Main results and context: }The most widely used solutions for the above problem~(detailed in Sec.~\ref{subsec:rel_work}) are parametric methods based on uniform sampling of nodes from the network~$G$ and, tend to produce \mbox{inaccurate~}(compared to the method we propose) results as a consequence of the large variance introduced by the heavy tail of the power-law degree distribution. As an alternative, \cite{eom2015tail} proposed a non-parametric estimation method called \textit{tail-scope}~(detailed in Sec.~\ref{subsec:rel_work}) based on a non-uniform sampling method motivated by the concept of the \textit{friendship paradox}. 

Since \cite{eom2015tail} is the closest idea to ours, let us briefly compare the results in \cite{eom2015tail} to our results.  This friendship paradox based sampling method can be thought of as a more accurate measurement sensor~(compared to the vanilla uniform sampling) that reduces the effect of measurement noise introduced by the heavy tail of the power-law degree distribution. However, the non-parametric method proposed in \cite{eom2015tail} does not exploit the explicit power-law form given in Eq.~\eqref{eq:pk_proportional_k} and also does not provide analytical guarantees for its accuracy~(e.g.~bias, variance and mean-squared error for a given sample size compared to the widely used methods). Further, it is also not clear how the method proposed in \cite{eom2015tail} generalizes to degree distributions that take forms other than power-law (e.g.~exponential degree distributions). As a solution we propose an estimator which exploits both the parametric form of power-law given in Eq.~\eqref{eq:pk_proportional_k} and the friendship paradox based sampling method. The proposed method is provably superior to the state of the art in both finite and asymptotic sample regimes and also possesses desirable statistical properties such as asymptotic unbiasedness, consistency and asymptotic normality. Further, the proposed method easily generalizes to other parametric forms of degree distributions that do not necessarily have heavy tails (e.g.~exponential degree distributions) while preserving all its desirable statistical properties. More specifically, this paper presents:
\begin{compactenum}
	\item A maximum likelihood estimation method that exploits the concept of \textit{friendship paradox} for sampling nodes according to a non-uniform distribution.
	
	\item Expressions for bias, variance and Cram\`{e}r-Rao lower bound of the proposed estimation method and their comparisons with alternative methods that prove how the proposed method outperforms the alternative methods.
	
	\item Numerical experiments that illustrate the better performance~(in terms of bias, variance mean-squared error) of the proposed method compared to alternative methods under various different conditions.
	
	\item Empirical results that verify the conclusions reached from the statistical analysis and numerical experiments using $18$ real-world network datasets.
	
	\item Generalization of the main results to exponential degree distributions to illustrate how all desirable statistical properties of the proposed method are valid even in non heavy-tailed degree distributions. 
\end{compactenum}
An overview of the key ideas of this paper along with their organization in different sections is shown in Fig.~\ref{fig:blockdiag} and is also outlined below in detail.

\vspace{0.25cm}
\noindent
{\bf Organization: } Sec.~\ref{subsec:fp} reviews the friendship paradox and Sec.~\ref{subsec:rel_work} discusses related literature. Sec.~\ref{sec:MLEs} presents our main idea of maximum likelihood estimation that exploits the friendship paradox based sampling. Sec.~\ref{sec:statistical_analysis} presents  the statistical analysis of the proposed estimate and compares it with widely used alternative methods. Sec.~\ref{sec:numerical_results} presents detailed numerical experiments that support and complement the statistical analysis. Sec.~\ref{sec:empirical_results} presents empirical results that compare the proposed method with alternative methods using several real-world datasets. Finally, Appendix~\ref{append:MLE_exponential} shows how the main results extend to other types of degree distributions and, Appendix~\ref{append:reproducibility} provides additional details on the numerical and empirical results. The code is publicly available on the Github repository~\cite{anonymousGitHubRepo2020KDD} to ensure that all results in the paper are completely reproducible.


\subsection{Friendship Paradox}
\label{subsec:fp}
The ``friendship paradox" is a form of observation bias first presented in~\cite{feld1991} by Scott L. Feld in 1991. The friendship paradox states, ``\textit{on average, the number of friends of a random friend is always greater than the number of friends of a random individual}". Formally:
\begin{theorem} { (Friendship Paradox \cite{feld1991})}
	\label{th:friendship_paradox_Feld}
	Consider an undirected graph~${G = (V,E)}$. Let $X$~be a node sampled uniformly from~$V$ and, $Y$~be a uniformly sampled end-node from a uniformly sampled edge~${e\in E}$. Then,
	\begin{equation}
	\mathbb{E} \{d(Y)\} \geq \mathbb{E}\{d(X)\},
	\end{equation} where, $d(X)$ and $d(Y)$ denote the degrees of $X$ and $Y$, respectively. 
\end{theorem}

In Theorem~\ref{th:friendship_paradox_Feld}, the random variable~$Y$ depicts a \textit{random friend}~(or a random neighbor) since it is obtained by sampling a pair of friends~(i.e.~an edge from the graph) uniformly and then choosing one of them via an unbiased coin flip. Equivalently, a random friend~$Y$ is a node sampled from~$V$ with a probability proportional to its degree. Further, the degree distribution~$\neighbordegdist$ of a random friend~$Y$ is given by
\begin{equation}
\label{eq:qk_proportional_kpk}
\neighbordegdist(k) \propto k\degdist(k),
\end{equation}
where, $\degdist$ is the degree distribution defined in Eq.~\eqref{eq:pk_proportional_k}. Eq.~(\ref{eq:qk_proportional_kpk}) follows from the fact that each node with a degree~$k$ appears as a friend of~$k$~other individuals. Note from~Eq.~\eqref{eq:qk_proportional_kpk} that the degree distribution~$\neighbordegdist$ of random friend~$Y$ is a right-skewed version of the degree distribution $\degdist$ and thus resulting in the friendship paradox in Theorem~\ref{th:friendship_paradox_Feld}.

The intuition behind Theorem \ref{th:friendship_paradox_Feld} is as follows. Individuals with large numbers of friends~(high degree) appear as the friends of a large number of individuals. Hence, such popular individuals can contribute to an increase in the average number of friends of friends. On the other hand, individuals with small numbers of friends~(low degree) appear as friends of a smaller number of individuals. Hence, they cannot cause a significant change in the average number of friends of friends. This asymmetric contribution of high and low degree individuals to the average number of friends of friends causes the friendship paradox. Further, \cite{cao2016} shows that the original version of the friendship paradox~(Theorem~\ref{th:friendship_paradox_Feld}) is a consequence of the monotone likelihood ratio stochastic ordering between random variables $d(Y)$ and $d(X)$.

\subsection{Motivation and Related Work \label{subsec:rel_work}}

Widely used methods for estimating the power-law exponent~$\alpha$ include:
\begin{compactenum}
	\item{\it Linear regression:} Using the empirical degree distribution~$\empiricaldegdist$ on double-logarithmic axes~(i.e.~$\ln \empiricaldegdist(k) $ with $\ln k$ where $\empiricaldegdist$ is the empirical distribution of degrees of $\samplesize$ uniformly sampled nodes), fit a straight line~(using least squares) whose slope is the estimate of~$\alpha$. The idea of this method is that $\ln \degdist (k)$ varies linearly with $\ln k$ with a slope of $-\alpha$ according to~Eq.~\eqref{eq:pk_proportional_k}.
	
	\item{\it Maximum likelihood estimation with uniformly sampled nodes:} Sample a set of nodes $X_1, X_2, \dots X_\samplesize$ independently and uniformly and, compute~$\vanillaMLE$ which maximizes the likelihood of observing the degree sequence~$d(X_1), d(X_2), \dots, d(X_n)$.
\end{compactenum}
Previous works~\cite{goldstein2004problems, bauke2007parameter} show that linear regression~(as well as its variants) for estimating the power-law exponent~$\alpha$ yields inaccurate results~(compared to maximum likelihood estimation method) due to two main reasons. First, the lack of data points from the tail of the power-law degree distribution~$\degdist$~(to construct the empirical degree distribution~$\empiricaldegdist$) systematically underestimates the power-law exponent~$\alpha$. Second, the log-log transformation of the empirical degree distribution violates several assumptions~(such as constant variance across all data points and zero-mean Gaussian noise) that are required to make the least squares estimate unbiased and statistically efficient~{(see~\cite{clauset2009power} for a detailed survey of these systematic inaccuracies in using linear regression method to estimate~$\alpha$)}. Therefore, the linear regression method for estimating power-laws is not well-justified from a statistical viewpoint. In comparison, maximum likelihood estimation with uniformly sampled nodes is a more principled approach which has been shown to achieve a better accuracy in estimating the power-law exponent~$\alpha$~\cite{clauset2009power, virkar2014power}. Maximum likelihood estimates also possess several appealing statistical properties including consistency, asymptotic unbiasedness and asymptotic efficiency. Hence, maximum likelihood estimation with uniformly sampled nodes is currently regarded as the state of the art method for estimating power-law degree distributions, and we use it as the main benchmark for evaluating the performance of the method proposed in this paper. 

\paragraph{Use of friendship paradox in estimation problems: } The friendship paradox has been used in several applications related to networks under the broad theme \textit{``how network biases can be exploited effectively in estimation problems?"}. For example, \cite{
	garcia2014using, 
	christakis2010} show how the friendship paradox can be used for quickly detecting a disease outbreak, \cite{nettasinghe2018your} proposes polling algorithms that exploits the friendship paradox for efficiently estimating the fraction of individuals with a certain attribute~(e.g.~intending to vote for a certain political party) in a social network. Apart from these, \cite{alipourfard2019friendship, jackson2019friendship, nettasinghe2019diffusion, 
	chin2018evaluating, 
	lee2019impact, 
	higham2018centrality, 
	bagrow2017friends, 
	lerman2016,
	momeni2016qualities, 
	lattanzi2015, eom2014, kooti2014network, hodas2013} 
also explore further generalizations (to directed graphs, attributes other than degree, etc.) and applications~(influence maximization etc.) of friendship paradox. 

More closely related to our work, \cite{eom2015tail} presents a method named \textit{tail-scope} which utilizes the friendship paradox for non-parametric estimation of heavy-tailed degree distributions. Tail-scope first obtains an empirical estimate~$\empiricalneighbordegdist$ of the neighbor degree distribution $\neighbordegdist$~(defined in Eq.~\eqref{eq:qk_proportional_kpk}) by sampling random neighbors~(denoted by random variable~$Y$ in Theorem~\ref{th:friendship_paradox_Feld}). Then, following Eq.~\eqref{eq:qk_proportional_kpk}, the empirical neighbor degree distribution~$\empiricalneighbordegdist(k)$ scaled by $k$ is used as the estimate of the degree distribution $\degdist(k)$~i.e.
\begin{equation}
\tailscope(k) \propto \frac{\empiricalneighbordegdist(k)}{k}.
\end{equation} The rationale behind tail-scope is that the empirical neighbor degree distribution~$\empiricalneighbordegdist$ will include more high degree nodes~(due to the friendship paradox) and hence, the scaled estimate~$\tailscope$ will capture the tail of the degree distribution better compared to the empirical degree distribution~$\empiricaldegdist$~(which is obtained with uniformly sampled nodes). While the method we propose is motivated in part by tail-scope, there are several key differences between tail-scope method and our method. Firstly, the method proposed in this paper is a parametric method that makes use of the specific power-law form in Eq.~\eqref{eq:pk_proportional_k} whereas tail-scope is a non-parametric method for general heavy-tailed degree distributions. Secondly, the method we propose has a provably better accuracy compared to the state of the art method~(maximum likelihood estimation with uniformly sampled nodes) and it also possesses desirable statistical properties including strong consistency, asymptotic unbiasedness and asymptotic statistical efficiency whereas such analytical guarantees are not available for tail-scope method. Finally, as we show in Appendix~\ref{append:MLE_exponential}, the proposed parametric method easily generalizes to other parametric forms of degree distributions~(e.g.~exponential distribution) as well, while preserving the performance guarantees (compared to the vanilla method based on uniformly sampled nodes) that it offers. Further, there have been several other works (e.g.~\cite{ribeiro2010estimating, ribeiro2012estimation}) which also suggest that edge sampling based methods (such as random walk methods) can offer better accuracy (compared to uniform sampling) in non-parametric estimation tasks by obtaining more samples from the tail of the degree distribution.


In summary, it has been shown in the literature~\cite{goldstein2004problems,bauke2007parameter,clauset2009power} that maximum likelihood methods are more suitable for estimating power-law degree distributions of the form in Eq.~\eqref{eq:pk_proportional_k} compared to alternative methods. Further, exploiting the friendship paradox~(Theorem~\ref{th:friendship_paradox_Feld}) has shown to be effective in empirical estimation of heavy-tailed degree distributions by including more high degree nodes into the sample~\cite{eom2015tail}. Motivated by these findings, this paper combines maximum likelihood estimation with friendship paradox based sampling in a principled manner to obtain an asymptotically unbiased, strongly consistent and statistically efficient estimate that provably outperforms the state of the art.

\section{Maximum likelihood estimation of power-law exponent}
\label{sec:MLEs}

This section introduces two different network sampling methods (uniform sampling and friendship paradox based sampling); we then compare the maximum likelihood estimate of the power-law exponent~$\alpha$ for these two sampling methods. The statistical analysis~(presented in Sec.~\ref{sec:statistical_analysis}) of the maximum likelihood estimates (MLEs) for the two sampling methods illustrates how the MLE obtained with friendship paradox based sampling~(proposed method)
outperforms the MLE obtained with uniform sampling (classically used method).


We first state the key assumptions used in deriving and analyzing the MLEs. 
\begin{assumption}
	\label{assumption:continuous_distribution}
	The power-law distribution $\degdist$ is continuous in $k$ and is of the form,
	\begin{align}
	\label{eq:power_law_distribution}
	p(k) =  \frac{\alpha - 1}{\mindegree}\bigg( \frac{k}{\mindegree}	\bigg)^{-\alpha}, \quad k\geq \mindegree
	\end{align}
	where, $\mindegree$ is the minimum degree.
\end{assumption}

\begin{assumption}
	\label{assumption:alpha_range}
	The power-law exponent~$\alpha$ is greater than~$2$ i.e.~${\alpha > 2}$.
\end{assumption}

Assumption~\ref{assumption:continuous_distribution} allows us to derive closed-form expressions for MLEs for the power-law exponent~$\alpha$. A similar assumption~(on continuity) has been used in \cite{eom2015tail} that deals with estimating heavy-tailed degree distributions. Further, Assumption~\ref{assumption:continuous_distribution} is naturally applicable for weighted networks where the weighted degree~(also called node strength) follows a continuous power-law distribution~\cite{barrat2004architecture, antoniou2008statistical}. For discrete power-law distributions, MLEs are not available in closed-form~\cite{clauset2009power}. Thus, Assumption~\ref{assumption:continuous_distribution} is useful in our theoretical analysis. In Sec.~\ref{subsec:discrete_MLEs}, we also briefly discuss how the MLEs for discrete power-law distributions can be approximated using closed-form expressions.


All moments~$m\geq \alpha - 1$ diverge for the power-law distribution given in~Eq.~\eqref{eq:power_law_distribution}. Hence, both mean and variance diverge when $\alpha \leq 2$ and variance diverges (and mean is finite) when $2<\alpha \leq 3$. Therefore, Assumption~\ref{assumption:alpha_range} ensures that the degree distribution~$\degdist$ in Eq.~\eqref{eq:power_law_distribution} has a finite mean which is necessary for the derivation of the proposed friendship paradox based MLE in Sec.~\ref{subsec:FP_MLE}.

\subsection{Vanilla MLE with uniform sampling}
\label{subsec:vanilla_MLE}
In the classical vanilla maximum likelihood estimation method~\cite{goldstein2004problems, clauset2009power}, $n$ number of nodes $X_1, X_2, \dots, X_{\samplesize}$ are independently and uniformly sampled from the network. Then, the likelihood of observing the degree sequence $d(X_1), d(X_2), \dots, d(X_{\samplesize})$ is,
\begin{align}
\mathbb{P}\{d(X_1), \dots, d(X_{\samplesize})\vert \alpha \} = \prod_{i = 1}^{\samplesize} \frac{\alpha - 1}{\mindegree}\bigg( \frac{d(X_i)}{\mindegree}	\bigg)^{-\alpha} \nonumber
\end{align}
following~Eq.~\eqref{eq:power_law_distribution}. Therefore, the log-likelihood for the vanilla method is,
\begin{align}
\begin{split}
\label{eq:loglikelihood_X}
\loglikelihoodX(\alpha) &= \ln \mathbb{P}\{d(X_1), d(X_2), \dots, d(X_{\samplesize})\vert \alpha \} \\ 
& =\samplesize\ln (\alpha - 1) - \samplesize\ln (\mindegree) - \alpha \sum_{i = 1}^{\samplesize} \ln \bigg(\frac{d(X_i)}{\mindegree}\bigg). 
\end{split}
\end{align}
Then, by solving $\frac{\partial \loglikelihoodX}{\partial \alpha} = 0$, we get the vanilla MLE of the power-law exponent $\alpha$ as,
\begin{align}
\label{eq:vanilla_MLE}
\vanillaMLE = \frac{\samplesize}{\sum_{i = 1}^{\samplesize}\ln\Big(\frac{d(X_i)}{\mindegree}\Big)} + 1.
\end{align}

Next, we present a maximum likelihood estimator that exploits the friendship paradox.
\subsection{MLE with friendship paradox based sampling}
\label{subsec:FP_MLE}
\paragraph{Neighbor Degree Distribution: }{Recall from Sec.~\ref{subsec:fp} that~$Y$ denotes a random neighbor~i.e.~uniformly sampled end-node of a uniformly sampled edge. The neighbor degree distribution~$q(k)$ is proportional to~$k\degdist(k)$ as stated in Eq.~\eqref{eq:qk_proportional_kpk}. Hence, the normalizing constant~$C_q$ of the neighbor degree distribution $q$ can be derived as,
	\begin{align}
	C_q &=  \int_{\mindegree}^{\infty}k\degdist(k)dk 
	= \int_{\mindegree}^{\infty}k \frac{\alpha - 1}{\mindegree}\bigg( \frac{k}{\mindegree}	\bigg)^{-\alpha}dk 
	= \frac{\alpha - 1}{\alpha - 2}\mindegree
	\end{align}
	where, $p(k)$ is the power-law degree distribution defined in Eq.~\eqref{eq:power_law_distribution} that is continuous in degree~$k$ by Assumption~\ref{assumption:continuous_distribution}. Note that the normalizing constant of the distribution $\neighbordegdist$ is equal to the first moment~(i.e.~mean) of the power-law degree distribution~$\degdist$~(defined in Eq.~\eqref{eq:power_law_distribution}) and, is guaranteed to exist by Assumption~\ref{assumption:alpha_range}. Then, it follows that,
	\begin{align}
	\neighbordegdist(k) = \frac{1}{C_q}k\degdist(k) 
	= \frac{\alpha - 2}{\mindegree}\bigg( \frac{k}{\mindegree}	\bigg)^{-(\alpha -1)}, \quad k\geq \mindegree. \label{eq:eq:power_law_neighbor_degree_distribution}
	\end{align}
}

\paragraph{Friendship Paradox based MLE: }
The friendship paradox based maximum likelihood estimator begins with sampling $\samplesize$ number of random neighbors $Y_1, Y_2,\dots, Y_{\samplesize}$ from the network independently. Then, the likelihood of observing the neighbor degree sequence $d(Y_1), d(Y_2), \dots, d(Y_{\samplesize})$ can be written using the neighbor degree distribution in Eq.~\eqref{eq:eq:power_law_neighbor_degree_distribution} as, 
\begin{equation}
\hspace{-0.0cm}\mathbb{P}\{d(Y_1), \dots, d(Y_{\samplesize})\vert \alpha \} = \prod_{i = 1}^{\samplesize} \frac{\alpha - 2}{\mindegree}\bigg( \frac{d(Y_i)}{\mindegree}	\bigg)^{-(\alpha - 1)}. \nonumber
\end{equation}
Therefore, the log-likelihood for the friendship paradox based sampling method is,
\begin{align}
\begin{split}
\label{eq:loglikelihood_Y}
\loglikelihoodY(\alpha) &= \ln \mathbb{P}\{d(Y_1), d(Y_2), \dots, d(Y_{\samplesize})\vert \alpha \} \\ 
& =\samplesize\ln (\alpha - 2) - \samplesize\ln (\mindegree) - (\alpha-1) \sum_{i = 1}^{\samplesize} \ln \bigg(\frac{d(Y_i)}{\mindegree}\bigg).
\end{split}
\end{align}
Then, by solving $\frac{\partial \loglikelihoodY}{\partial \alpha} = 0$, we get the friendship paradox based MLE of the power-law exponent~$\alpha$ as,
\begin{align}
\label{eq:FP_MLE}
\fpMLE = \frac{\samplesize}{\sum_{i = 1}^{\samplesize}\ln\Big(\frac{d(Y_i)}{\mindegree}\Big)} + 2.
\end{align}

\begin{remark}[Power-law form of the neighbor degree distribution] \normalfont
	\label{remark:powerlaw_form_of_q}
Note that the neighbor degree distribution $q(k)$ (given in Eq.~\eqref{eq:eq:power_law_neighbor_degree_distribution}) is also a power-law degree distribution. More specifically, the neighbor degree distribution $q(k)$ can be obtained from the power-law degree distribution $p(k)$ (given in Eq.~\eqref{eq:power_law_distribution}) by simply replacing the power-law exponent $\alpha$ with $\alpha - 1$. Thus, the neighbor degree-distribution $q(k)$ is also a power-law but, with a heavier tail (i.e.~a smaller power-law exponent) compared to the degree distribution $p(k)$. This is due to the fact that $q(k) \propto kp(k)$ where the scaling term $k$ reduces the power-law exponent of $p(k)$ by a value of $1$.
\end{remark}

\begin{remark}[Alternative implementations] \normalfont
	Recall~(from Sec.~\ref{subsec:fp}) 
	that a random neighbor~$Y$ is a uniformly sampled end node of a uniformly sampled edge~${e \in E}$. In applications such as online social networks, uniform link sampling~(and therefore, sampling random neighbors) is possible since each edge has a unique integer ID assigned from a specific range of integers~\cite{leskovec2006}. In applications where sampling uniform edges is not possible~(e.g.~unknown network, lack of edge IDs), one possible method to sample random neighbors~(to implement the friendship paradox based MLE in Eq.~\eqref{eq:FP_MLE}) is by using random walks. Assuming the underlying network $G = (V,E)$ is a connected, non-bipartite graph, the stationary distribution of a random walk on~$G$ samples each node $v \in V$ with a probability proportional to the degree~$d(v)$ of node~$v$~(page 298,~\cite{durrett2010_probability})~i.e.~the stationary distribution of a random walk on a connected, non-bipartite graph samples random neighbors. Hence, a node sampled from a sufficiently long random walk has approximately the same distribution as a random neighbor~$Y$. Another possibility is to use a second version of the friendship paradox which states ``uniformly sampled friend of a uniformly sampled node has more friends than a uniformly sampled node, on average''~(see \cite{cao2016, chin2018evaluating} for more details on this second version). This second version does not require sampling links and, is equivalent to a one step random walk. 
\end{remark}

\subsection{Approximating the maximum likelihood estimates for discrete data}
\label{subsec:discrete_MLEs}

In Sec.~\ref{subsec:vanilla_MLE} and Sec.~\ref{subsec:FP_MLE}, we derived the MLEs for the power-law exponent~$\alpha$ assuming that the observations are from a continuous power-law distribution~(Assumption~\ref{assumption:continuous_distribution}). However, network degree distributions are discrete distributions for which a closed-form MLE for the power-law exponent is not available~(as discussed at the beginning of Sec.~\ref{sec:MLEs}). Therefore, the maximum likelihood estimates of the power-law exponent for such discrete distributions have to be obtained by numerical methods. Two solutions are available to overcome the lack of a closed-form expression for the MLE of the power-law exponent of a discrete power-law degree distribution.

The first solution is to simply apply the expressions derived under the Assumption~\ref{assumption:continuous_distribution} (in Sec.~\ref{subsec:vanilla_MLE} and Sec.~\ref{subsec:FP_MLE}) to the discrete case~i.e.~pretending that the discrete data is from a continuous distribution for the estimation purpose. The second solution is to use an approximate closed-form solution for the power-law exponent of a discrete power-law distribution as proposed in \cite{clauset2009power}. The approximate MLE of the power-law exponent $\alpha$ for the discrete case under the vanilla uniform sampling is:
\begin{equation}
\label{eq:vanilla_discrete_MLE}
\vanillaDiscreteMLE = \frac{\samplesize}{\sum_{i = 1}^{\samplesize}\ln\Big(\frac{d(X_i)}{\mindegree - 0.5}\Big)} + 1
\end{equation}
where, $X_1, X_2, \dots, X_{\samplesize}$ are independently and uniformly sampled nodes~\cite{clauset2009power}. Since the neighbor degree distribution $\neighbordegdist$ of a network with a power-law degree distribution is also a power-law (as stated in Remark~\ref{remark:powerlaw_form_of_q}), the approximate MLE of the power-law exponent $\alpha$ for the discrete case under the friendship paradox based sampling is:
\begin{equation}
\label{eq:FP_discrete_MLE}
\fpDiscreteMLE = \frac{\samplesize}{\sum_{i = 1}^{\samplesize}\ln\Big(\frac{d(Y_i)}{\mindegree - 0.5}\Big)} + 2
\end{equation}
where, $Y_1, Y_2,\dots, Y_{\samplesize}$ are  independently sampled random neighbors. Henceforth, we refer to~$\vanillaDiscreteMLE$~(in Eq.~\eqref{eq:vanilla_discrete_MLE}) and~$\fpDiscreteMLE$~(in Eq.~\eqref{eq:FP_discrete_MLE}) as vanilla discrete MLE and friendship paradox based discrete MLE, respectively.

In Sec.~\ref{sec:numerical_results} and Sec.~\ref{sec:empirical_results}, we will explore the performance vanilla discrete MLE~$\vanillaDiscreteMLE$ and friendship paradox based discrete MLE~$\fpDiscreteMLE$ in terms of bias, variance and mean-squared error using synthetic and real-world network datasets. Further, we will also compare their performance with the MLEs in Eq.~\eqref{eq:vanilla_MLE} and Eq.~\eqref{eq:FP_MLE} that were derived assuming a continuous power-law distribution.

\section{Statistical analysis of the Estimates}
\label{sec:statistical_analysis}
This section presents the statistical analysis of the two estimates: vanilla MLE~$\vanillaMLE$ in Eq.~\eqref{eq:vanilla_MLE}~(which uses uniform sampling of nodes) and the friendship paradox based MLE~$\fpMLE$ in Eq.~\eqref{eq:FP_MLE}~(which samples nodes non-uniformly according to the friendship paradox). The statistical analysis below shows that the proposed friendship paradox based MLE~$\fpMLE$ outperforms the widely used vanilla MLE~$\vanillaMLE$ in terms of bias, variance and Cram\`{e}r-Rao Lower Bound.

\vspace{0.1cm}
\paragraph{Comparison of bias and variance of MLEs for finite sample size: } The following result from~\cite{muniruzzaman1957measures}~(also discussed in \cite{clauset2009power} in a broader context) characterizes the bias and variance of the vanilla MLE~$\vanillaMLE$ under finite sample sizes~$\samplesize < \infty$.

\vspace{-0.1cm}
\begin{theorem}[Bias and variance of the vanilla MLE~\cite{muniruzzaman1957measures, clauset2009power}]
	\label{th:bias_variance_vanillaMLE}
	The bias and variance of the vanilla MLE~$\vanillaMLE$ in Eq.~\eqref{eq:vanilla_MLE} for sample size~$\samplesize$ are given by,
	\begin{align}
	\begin{split}
	\hspace{-0.0cm}\bias\{\vanillaMLE\} &= \frac{\alpha - 1}{\samplesize -1 },\;\; \textbf{ \normalfont for } \samplesize > 1\\  \var \{\vanillaMLE\} &= \frac{\samplesize^2(\alpha - 1)^2}{(\samplesize-1)^2(\samplesize-2)} ,\;\; \textbf{ \normalfont for } \samplesize > 2.
	\end{split}
	\label{eq:bias_variance_vanillaMLE}
	\end{align}
\end{theorem}

The bias and variance of the proposed friendship paradox based MLE~$\fpMLE$ under a finite sample size $\samplesize < \infty$ are characterized  in the following result,  allowing it to be compared with the vanilla MLE~$\vanillaMLE$. 

\begin{theorem}[Bias and variance of the friendship paradox based MLE]
	\label{th:bias_variance_fpMLE}
	The bias and variance of the friendship paradox based MLE~$\fpMLE$ in Eq.~\eqref{eq:FP_MLE} for sample size~$\samplesize$ are given by,
	\begin{align}
	\begin{split}
	\bias\{\fpMLE\} &= \frac{\alpha - 2}{\samplesize -1},\;\; \textbf{ \normalfont for } \samplesize > 1\\
	\var \{\fpMLE\} &= \frac{\samplesize^2(\alpha - 2)^2}{(\samplesize-1)^2(\samplesize-2)},\;\; \textbf{ \normalfont for } \samplesize > 2.
	\label{eq:bias_variance_FPmle}
	\end{split}
	\end{align}
\end{theorem}
\begin{proof}
Recall from Remark~\ref{remark:powerlaw_form_of_q} that the neighbor degree distribution $q(k)$ is also a power-law distribution with power-law exponent $\alpha - 1$. Thus, Theorem~\ref{th:bias_variance_fpMLE} can be obtained by replacing $\alpha$ in  Theorem~\ref{th:bias_variance_vanillaMLE} with $\alpha - 1$. The formal full proof is provided in Appendix~{\ref{append:bias_variance_fpMLE}}.
\end{proof}

The following immediate consequence of Theorem~\ref{th:bias_variance_vanillaMLE} and Theorem~\ref{th:bias_variance_fpMLE} shows that the proposed friendship paradox based MLE~$\fpMLE$ outperforms vanilla MLE~$\vanillaMLE$ for any finite sample size~$\samplesize$.

\begin{figure*}[]
	\centering
	\begin{subfigure}[!h]{0.47\textwidth}
		\centering
		\includegraphics[width= \textwidth, trim={0.2cm 0.25cm 0.25cm 0.2cm}, clip]{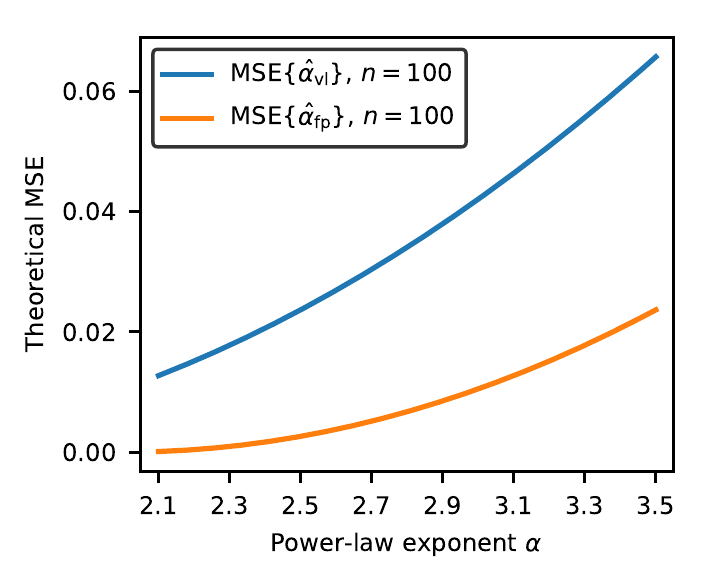}
		\caption{}
		\label{subfig:theoretical_MSEs}
	\end{subfigure}\hfill
	\begin{subfigure}[!h]{0.47\textwidth}
		\centering
		\includegraphics[width= \textwidth, trim={0.2cm 0.25cm 0.25cm 0.2cm}, clip]{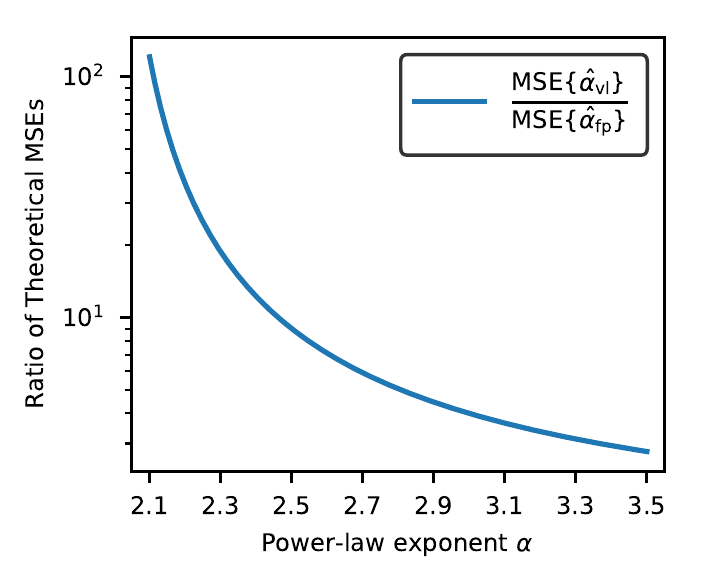}
		\caption{}
		\label{subfig:theoretical_MSE_ratios}
	\end{subfigure}\hfill 
	\caption{Theoretical mean-squared error (MSE) values of the vanilla estimate $\vanillaMLE$ and the friendship paradox based estimate $\fpMLE$ (Fig.~\ref{subfig:theoretical_MSEs}) and their ratios (Fig.~\ref{subfig:theoretical_MSE_ratios}) calculated analytically using Eq.~\eqref{eq:bias_variance_vanillaMLE} and Eq.~\eqref{eq:bias_variance_FPmle} for sample size $\samplesize = 100$. The plots show that, theoretically, the proposed friendship paradox based estimate $\fpMLE$ has a MSE which is smaller than the widely used vanilla MLE $\vanillaMLE$. The ratio between the MSEs increases as $\alpha$ approaches $2$.}
	\label{fig:theoretical_MSEs_and_ratios}
\end{figure*}

\begin{corollary}
	\label{cor:bias_variance_ordering}
	The bias and variance of the vanilla MLE~$\vanillaMLE$ defined in Eq.~\eqref{eq:vanilla_MLE} and the friendship paradox based MLE~$\fpMLE$ defined in Eq.~\eqref{eq:FP_MLE} satisfy,
	\begin{align}
	\begin{split}
	\bias\{\fpMLE\} &< \bias\{\vanillaMLE\},\;\;\textbf{ \normalfont for } \samplesize > 1\\
	\var\{\fpMLE\} &< \var\{\vanillaMLE\},\;\; \textbf{ \normalfont for } \samplesize > 2.
	\end{split}
	\end{align}
\end{corollary}

From Corollary~\ref{cor:bias_variance_ordering}, it immediately follows that the mean-squared error (MSE) of the proposed friendship paradox based MLE~$\fpMLE$  is also smaller compared to the MSE of the vanilla MLE~$\vanillaMLE$. Fig.~\ref{fig:theoretical_MSEs_and_ratios} illustrates the MSE values of the two estimates calculated using Eq.~\eqref{eq:bias_variance_vanillaMLE} and Eq.~\eqref{eq:bias_variance_FPmle} and the ratio between them for $\samplesize = 100$. It can be seen that the the ratio between the MSEs of the two methods grows rapidly to several orders of magnitude as the power-law exponent~$\alpha$ becomes smaller since both bias and variance of the friendship paradox based MLE~$\fpMLE$ approaches 0 as $\alpha$ goes to $2$.

Having established that the proposed friendship paradox based MLE~$\fpMLE$ outperforms the vanilla MLE~$\fpMLE$  for all sample sizes ${\samplesize < \infty}$, we now turn to the case where the sample size~$\samplesize$ tends to infinity. 

\vspace{0.1cm}
\paragraph{Comparison of the asymptotic properties of the MLEs: }
Based on Eq.~\eqref{eq:bias_variance_vanillaMLE},~Eq.~\eqref{eq:bias_variance_FPmle} and the strong law of large numbers, we have the following result:
\begin{theorem}[Asymptotic unbiasedness and strong consistency of MLEs]
	\label{th:strong_consistency}
	The vanilla MLE~$\vanillaMLE$~(defined in Eq.~\eqref{eq:vanilla_MLE}) and the friendship paradox based MLE~$\fpMLE$~(defined in Eq.~\eqref{eq:FP_MLE}) are asymptotically unbiased and strongly consistent~i.e.~converge to the true power-law exponent~$\alpha$ with probability $1$ as the sample size $\samplesize$ tends to infinity. 
\end{theorem} 
\begin{proof}
See Appendix~\ref{append:strong_consistency}.
\end{proof}

In order to analytically compare the two MLEs in the asymptotic regime, we use the  Cram\`{e}r-Rao Lower Bound (CRLB).  For a scalar random variable~(which is the case we deal with), CRLB simplifies to the reciprocal of the Fisher Information~\cite{casella2002statistical}. 
We stress that the CRLB is important because  the variance of any unbiased estimate~(or asymptotically unbiased estimate) is bounded below by the CRLB~i.e.~CRLB is the smallest variance achievable by any unbiased estimate. The following result characterizes the two Cram\`{e}r-Rao Lower Bounds,~$\CRLBvanilla$ and $\CRLBfp$, of the two estimates~$\vanillaMLE$ and~$\fpMLE$.

\begin{theorem}
	\label{th:CRLB}
	The  Cram\`{e}r-Rao Lower Bounds of the vanilla MLE $\vanillaMLE$ in Eq.~\eqref{eq:vanilla_MLE} and the friendship paradox based MLE~$\fpMLE$ in Eq.~\eqref{eq:FP_MLE} are given by,
	\begin{align}
	\hspace{-0.25cm}\CRLBvanilla(\alpha) = \frac{(\alpha - 1)^2}{\samplesize}, \quad	\CRLBfp(\alpha) = \frac{(\alpha - 2)^2}{\samplesize}.
	\end{align}
\end{theorem}
\begin{proof}
See Appendix~\ref{append:CRLB}. 
\end{proof}

Since maximum likelihood estimates are asymptotically normal and efficient~(achieves the CRLB), it follows that, 
\begin{align}
\begin{split}
\label{eq:asymptotically_normal_distributions}
\sqrt{\samplesize}(\vanillaMLE - \alpha) &\xrightarrow{d} \mathcal{N}(0, (\alpha - 1)^2) \\ \sqrt{\samplesize}(\fpMLE - \alpha) &\xrightarrow{d} \mathcal{N}(0, (\alpha - 2)^2).
\end{split}
\end{align}
Hence, it can be seen that the asymptotic variance of~$	\sqrt{\samplesize}(\fpMLE - \alpha)$ is smaller than that of~$	\sqrt{\samplesize}(\vanillaMLE - \alpha)$, implying that the proposed friendship paradox based MLE~$\fpMLE$ outperforms the vanilla MLE~$\vanillaMLE$ for large samples sizes.  

\begin{remark}[Exponential degree distributions] \normalfont Apart from the power-law degree distribution, several real world networks (e.g.~Worldwide Marine Transportation Network~\cite{wei2009worldwide}) have exponential degree distributions~i.e.~probability~$\degdist_{exp}(k)$ that a uniformly sampled node has~$k$~neighbors is proportional to~$e^{-\lambda k}$ where $\lambda>0$ is the fixed rate parameter~\cite{deng2011exponential}. The results in Sec.~\ref{sec:MLEs} and Sec.~\ref{sec:statistical_analysis} also extend to the case where the underlying network has an exponential degree distribution as shown in Appendix~\ref{append:MLE_exponential}. This result is important in that it shows that the proposed friendship paradox based maximum likelihood estimation method outperforms the vanilla method (based on uniform sampling) in the setting of light-tailed degree distributions as well. 
\end{remark}

\vspace{0.25cm}
\noindent
{\bf Summary of Statistical Analysis:} The statistical analysis motivates the use of the proposed friendship paradox based MLE~$\fpMLE$~(defined in Eq.~\eqref{eq:FP_MLE}) in place of the widely used vanilla MLE~$\vanillaMLE$~(defined in Eq.~\eqref{eq:vanilla_MLE}) for estimating power-law degree distributions. Theorems~\ref{th:bias_variance_vanillaMLE},~\ref{th:bias_variance_fpMLE} characterize the bias and variance of the two estimates $\vanillaMLE, \fpMLE$ and, Corollary~\ref{cor:bias_variance_ordering} concludes that the proposed method has a smaller bias and a smaller variance under finite sample size $\samplesize<\infty$. Then, Theorem~\ref{th:CRLB} gives the Cram\`{e}r-Rao bounds for the two estimates $\vanillaMLE, \fpMLE$ to show that the proposed MLE $\fpMLE$ outperforms the vanilla MLE $\vanillaMLE$ in the asymptotic regime and, possesses properties such as asymptotic unbiasedness, consistency and asymptotic normality. Thus, the statistical analysis concludes that the proposed friendship paradox based MLE~$\fpMLE$ outperforms the widely used vanilla MLE~$\vanillaMLE$ in both finite and asymptotic regimes.

\section{Numerical Comparison of the Performance of the Estimates}
\label{sec:numerical_results}

This section presents detailed numerical examples that complement the statistical analysis in Sec.~\ref{sec:statistical_analysis}. More specifically, we: 
\begin{compactenum}[i.]
	\item compare the four estimates of the power-law exponent~$\alpha$~(presented in Sec.~\ref{sec:MLEs}): vanilla maximum likelihood estimate~(MLE)~${\vanillaMLE}$~(widely used method) given in Eq.~\eqref{eq:vanilla_MLE},  friendship paradox based MLE~$\fpMLE$~(proposed method) given in Eq.~\eqref{eq:FP_MLE} and their discrete versions~$\vanillaDiscreteMLE, \fpDiscreteMLE$ given in Eq.~\eqref{eq:vanilla_discrete_MLE} and Eq.~\eqref{eq:FP_discrete_MLE} respectively, in terms of empirically estimated bias, variance, MSE and the theoretical CRLB values
	
	\item  numerically evaluate the effect of the sample size $\samplesize$ on the performance of the friendship paradox based methods~($\fpMLE, \fpDiscreteMLE$) in comparison to the widely used vanilla methods based on uniform sampling~($\vanillaMLE, \vanillaDiscreteMLE$).
\end{compactenum}

\vspace{0.25cm}
\noindent
{\bf Simulation Setup: } We use the configuration model~\cite{newman2003structure} to generate networks with $50000$ nodes that have degree sequences sampled from a power-law distribution with a given power-law exponent~$\alpha$ and a minimum degree~$\mindegree$.  Then, the bias, variance and MSE of the estimates (for power-law exponent~$\alpha$ in the interval~$[2.1,\, 3.5]$ and minimum degree~$\mindegree = 1,\,5,\,50,\,500$) are empirically estimated using a Monte-Carlo simulation over  {$5000$} independent iterations. Additional details on the simulation setup are provided in the Appendix~\ref{append:reproducibility} and all the simulation codes are publicly available in the Github repository~\cite{anonymousGitHubRepo2020KDD}.

\paragraph{Numerical Comparison of bias, variance and mean-squared error: }
\begin{figure}
	\includegraphics[width=\columnwidth,  trim={0.3cm 0.3cm 0.3cm 0.3cm},clip]{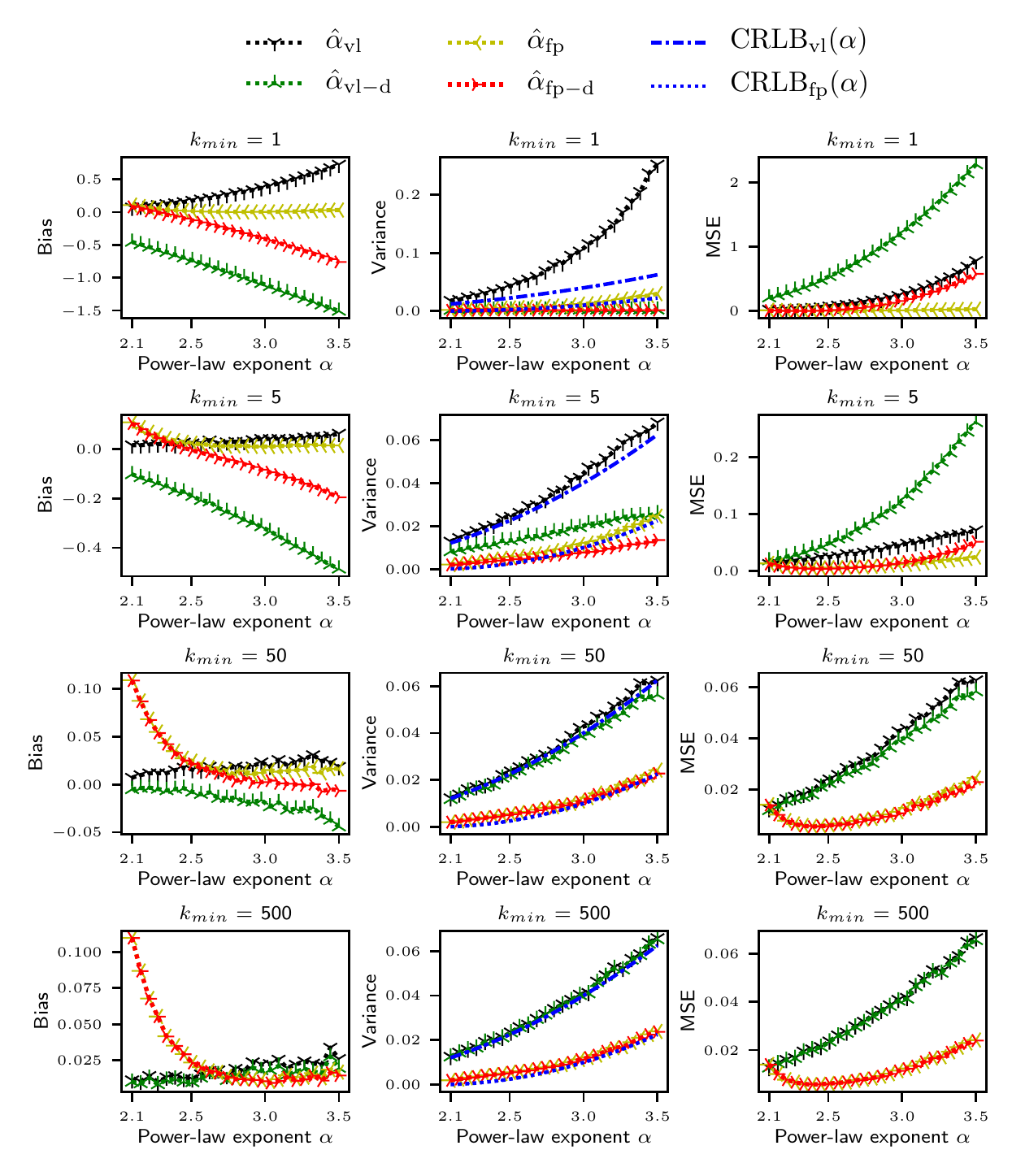}
	\caption{ 
		Empirical bias, variance, mean-squared error~(MSE) and theoretical Cram\`{e}r-Rao lower bound~(CRLB) values of the four estimates~(presented in Sec.~\ref{sec:MLEs}): vanilla maximum likelihood estimate~(MLE)~${\vanillaMLE}$ given in Eq.~\eqref{eq:vanilla_MLE},  friendship paradox based MLE~$\fpMLE$ given in Eq.~\eqref{eq:FP_MLE} and their discrete versions~$\vanillaDiscreteMLE, \fpDiscreteMLE$ given in Eq.~\eqref{eq:vanilla_discrete_MLE} and Eq.~\eqref{eq:FP_discrete_MLE} respectively, for sample size ${\samplesize = 100}$. 
The plots show that the proposed friendship paradox based MLEs $\fpMLE, \fpDiscreteMLE$ outperform the vanilla MLEs~$\vanillaMLE, \vanillaDiscreteMLE$ in terms of MSE for all considered values of power-law exponent~$\alpha$ and minimum-degree $\mindegree$. Further, the variance of the friendship paradox based MLEs are approximately equal to the best achievable variance (CRLB) under the friendship paradox based sampling and they also induce a smaller bias compared to the vanilla methods. Thus, the figure suggests that the friendship paradox based MLEs yield better results compared to the widely used vanilla methods based on uniform sampling for estimating power-law degree distributions.
	}
	\label{fig:Var_MSE_CRLB}
\end{figure}

Fig.~\ref{fig:Var_MSE_CRLB} shows the empirically estimated bias, variance, MSE and the CRLB values of the four estimates of the power-law exponent~$\alpha$ presented in Sec.~\ref{sec:MLEs}: vanilla MLE~${\vanillaMLE}$~(widely used method) given in Eq.~\eqref{eq:vanilla_MLE},  friendship paradox based MLE~$\fpMLE$~(proposed method) given in Eq.~\eqref{eq:FP_MLE} and their discrete versions~$\vanillaDiscreteMLE, \fpDiscreteMLE$ given in Eq.~\eqref{eq:vanilla_discrete_MLE} and Eq.~\eqref{eq:FP_discrete_MLE} respectively, for sample size~$\samplesize = 100$. Several important observations that complement the statistical analysis in Sec.~\ref{sec:statistical_analysis} can be drawn from Fig.~\ref{fig:Var_MSE_CRLB}. 

First, the MSE~(shown in the 3\textsuperscript{rd} column of Fig.~\ref{fig:Var_MSE_CRLB}) of the friendship paradox based MLE~$\fpMLE$~(indicated in yellow) and its discrete version~$\fpDiscreteMLE$~(indicated in red) are smaller compared to the widely used methods that are based on uniform sampling~(i.e.~vanilla MLE~$\vanillaMLE$ and its discrete version~$\vanillaDiscreteMLE$). Further, the friendship paradox based MLE~$\fpMLE$~(indicated in yellow) has a slightly smaller MSE compared to its discrete version~$\fpDiscreteMLE$~(indicated in red), especially for smaller minimum degree values~(i.e.~$\mindegree \leq 5$) and larger power-law exponent values~(i.e.~$\alpha >3$). Thus, the numerical results suggest that the friendship paradox based MLEs~$\fpMLE, \fpDiscreteMLE$ produce smaller overall errors~(i.e.~MSE values) compared to the currently used methods that are based on uniform sampling~$\vanillaMLE, \vanillaDiscreteMLE$ and, verify the conclusions reached by the statistical analysis in Sec.~\ref{sec:statistical_analysis}.

Second, note that the empirical variance~(shown in the 2\textsuperscript{nd} column of Fig.~\ref{fig:Var_MSE_CRLB}) of the proposed friendship paradox based MLE~$\fpMLE$~(indicated in yellow) and its discrete version~$\fpDiscreteMLE$~(indicated in red) are smaller compared to the vanilla MLEs based on uniform sampling~(i.e.~vanilla MLE~$\vanillaMLE$ and its discrete version~$\vanillaDiscreteMLE$). Importantly, it can be observed that the variance of the proposed friendship paradox based MLE~$\fpMLE$ is approximately equal to the theoretical CRLB value~($\CRLBfp$~-~blue dotted line) for all considered values of minimum degree~$\mindegree$ and power-law exponent~$\alpha$. This observation indicates that the variance of the proposed friendship paradox based MLE~$\fpMLE$ is almost equal to the best achievable variance under the friendship paradox based sampling (since CRLB is the smallest achievable variance by any unbiased estimate of a parameter for a given sampling method as discussed in~Sec.~\ref{sec:statistical_analysis}). In comparison, variance of the vanilla MLE~$\vanillaMLE$ achieves the CRLB~($\CRLBvanilla$~-~blued dashed line) only for larger minimum degree values~($\mindegree > 5$) and, is still larger compared to the friendship paradox based MLE~$\fpMLE$. In both vanilla and friendship paradox based cases, the discrete MLEs~$\vanillaDiscreteMLE, \fpDiscreteMLE$ reduce the variance for smaller minimum degree values~($\mindegree < 5$) while still maintaining the trend that friendship paradox based discrete MLE has a smaller variance compared to the vanilla discrete MLE~(i.e.~$\var\{\fpDiscreteMLE\} < \var\{\vanillaDiscreteMLE\} $). Thus, the numerical comparison of the estimates based on variance also suggests that the friendship paradox based MLE~$\fpMLE$ and its discrete version~$\fpDiscreteMLE$ are better compared to the currently used methods that are based on uniform sampling~$\vanillaMLE, \vanillaDiscreteMLE$ . This numerical observation also agrees with the theoretical results~(Corollary~\ref{cor:bias_variance_ordering}).

Third, it can be observed~(from the 1\textsuperscript{st} column of Fig.~\ref{fig:Var_MSE_CRLB}) that the bias of the proposed friendship paradox based MLE~$\fpMLE$ is negligible and also smaller compared to that of the vanilla MLE $\vanillaMLE$ when the power-law exponent~$\alpha$ is larger than $2.5$ (for all considered values of the minimum degree~$\mindegree$). This agrees with the theoretical conclusions~(Corollary~\ref{cor:bias_variance_ordering}). However, it can be seen that the bias of the proposed estimate~$\fpMLE$ starts to slightly increase when the power-law exponent $\alpha$ decreases below approximately $2.5$  (especially for values of the minimum degree $\mindegree$ larger than $5$).  This does not affect the conclusion that the friendship paradox based methods outperform the vanilla methods based on uniform sampling because the friendship paradox based methods still have a smaller MSE as seem from the third column of Fig.~\ref{fig:Var_MSE_CRLB}. Further, this slight increase of the bias contradicts the statistical analysis~(Theorem~\ref{th:bias_variance_fpMLE}) which suggests that the bias of the friendship paradox based MLE~$\fpMLE$ should be a monotone function of the power-law exponent~$\alpha$ and should not depend on the minimum degree~$\mindegree$. We believe that this disparity between the theoretical analysis and the numerical analysis for small values of the power-law exponent $\alpha$ (i.e.~$\alpha < 2.5$) could be due to the bias introduced when representing a heavy-tailed power-law degree distribution with a finite graph in our numerical experiments.

\paragraph{Effect of the sample size $\samplesize$: }
\begin{figure}
	\includegraphics[width=\columnwidth,  trim={0.2cm 0.2cm 0.15cm 0.75cm},clip]{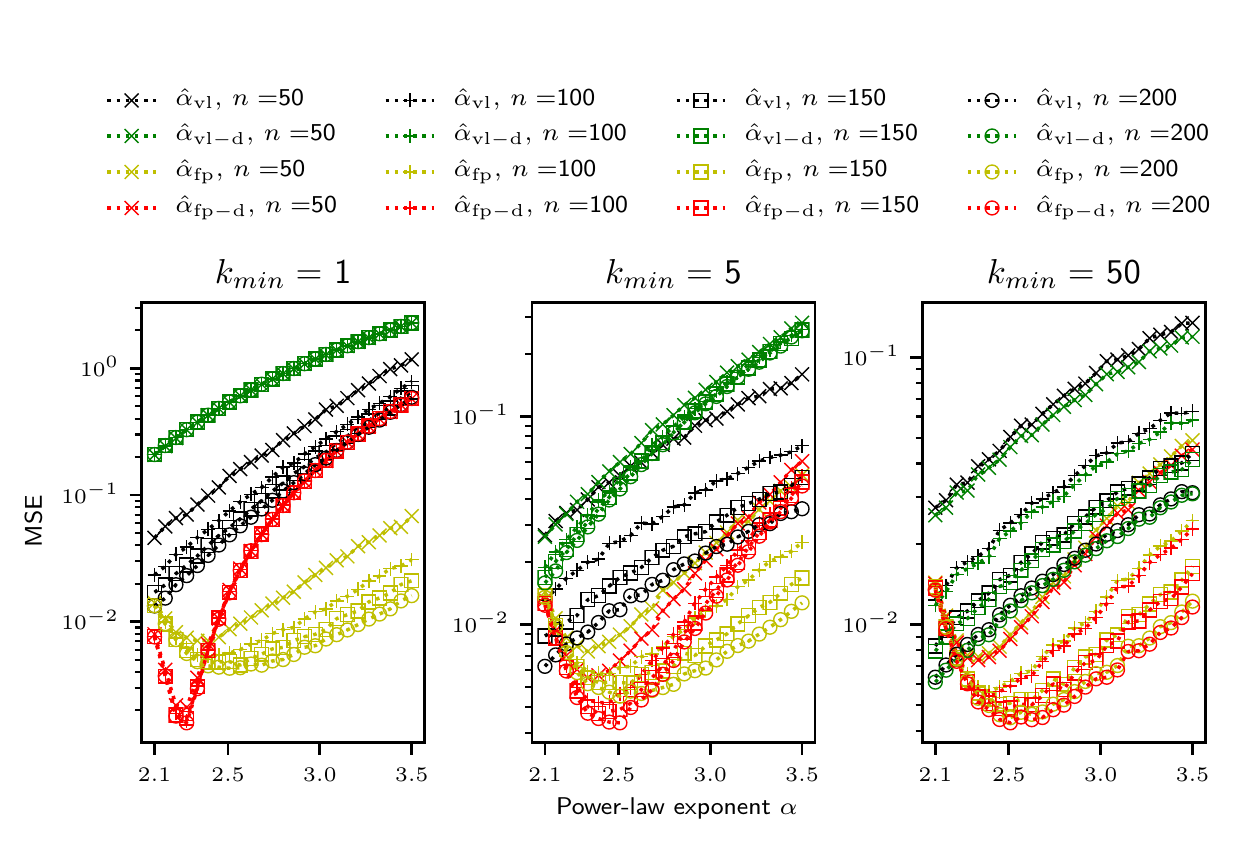}
	\caption{ 
		Empirical mean-squared error~(MSE) of the four estimates~(presented in Sec.~\ref{sec:MLEs}): vanilla maximum likelihood estimate~(MLE)~${\vanillaMLE}$ given in Eq.~\eqref{eq:vanilla_MLE},  friendship paradox based MLE~$\fpMLE$ given in Eq.~\eqref{eq:FP_MLE} and their discrete versions~$\vanillaDiscreteMLE, \fpDiscreteMLE$ given in Eq.~\eqref{eq:vanilla_discrete_MLE} and Eq.~\eqref{eq:FP_discrete_MLE} respectively, for different values of the sample size~${\samplesize}$~(recall that Fig.~\ref{fig:Var_MSE_CRLB} was for ${\samplesize = 100}$). The plots show that the proposed friendship paradox based MLEs~$\fpMLE, \fpDiscreteMLE$ achieve smaller MSEs with only a smaller sample size $\samplesize$ compared to the widely used uniform sampling based vanilla MLEs~$\vanillaMLE, \vanillaDiscreteMLE$ and thus, confirm the conclusions reached in the theoretical analysis~(Sec.~\ref{sec:statistical_analysis}).
	}
	\label{fig:SampleSize_vs_MSE}
\end{figure}
 Fig.~\ref{fig:SampleSize_vs_MSE} illustrates the MSE of the four estimates of the power-law exponent~$\alpha$ presented in Sec.~\ref{sec:MLEs}: vanilla MLE~${\vanillaMLE}$ given in Eq.~\eqref{eq:vanilla_MLE},  friendship paradox based MLE~$\fpMLE$ given in Eq.~\eqref{eq:FP_MLE} and their discrete versions~$\vanillaDiscreteMLE, \fpDiscreteMLE$ given in Eq.~\eqref{eq:vanilla_discrete_MLE} and Eq.~\eqref{eq:FP_discrete_MLE} respectively, for four values of the sample size $\samplesize = 50, 100, 150, 200$. Note that the friendship paradox based MLE~$\fpMLE$ and its discrete version~$\fpDiscreteMLE$ have smaller MSEs compared to the widely used vanilla methods based on uniform sampling~$\fpMLE, \fpDiscreteMLE$ for all considered values of the sample size. Further,~Fig.~\ref{fig:SampleSize_vs_MSE} shows that the friendship paradox based MLE~$\fpMLE$ and its discrete version~$\fpDiscreteMLE$ have approximately equal MSEs for larger values of the minimum degree~($\mindegree\geq 50$). The slight increase~(approximately $0.01$) of the MSE of the friendship paradox based methods~$\fpMLE, \fpDiscreteMLE$ when the power-law exponent $\alpha$ decreases below $2.3$ could be due to the bias introduced when representing a power-law degree distribution with a finite graph in numerical experiments~(as also discussed using the Fig.~\ref{fig:Var_MSE_CRLB} earlier). Thus, Fig.~\ref{fig:SampleSize_vs_MSE} further verifies the conclusions drawn from both the statistical analysis in Sec.~\ref{sec:statistical_analysis} and the observations from Fig.~\ref{fig:Var_MSE_CRLB} regarding the better performance of the friendship paradox based MLE~$\fpMLE$ and its discrete version~$\fpDiscreteMLE$ compared to the vanilla methods based on uniform sampling~$\vanillaMLE, \vanillaDiscreteMLE$. 

\vspace{0.25cm}
\noindent
{\bf Summary of Numerical Comparison: } This section presented numerical results that complement and support the statistical analysis~(Sec.~\ref{sec:statistical_analysis}) of the estimates. Specifically, it was shown that the empirical MSE of the proposed friendship paradox based MLE~$\fpMLE$ and its discrete version~$\fpDiscreteMLE$ have smaller MSEs compared to the currently used vanilla methods that are based on uniform sampling~$\vanillaMLE, \vanillaDiscreteMLE$ for various values of the power-law exponent~$\alpha$, the minimum degree~$\mindegree$ and the sample size~$\samplesize$. Further, the variance of the friendship paradox based MLE~$\fpMLE$ and its discrete version~$\fpDiscreteMLE$ are approximately equal to the best achievable variance by any unbiased estimate under friendship paradox based sampling~($\CRLBfp$) which is smaller compared to the best achievable variance by any unbiased estimate under uniform sampling sampling~($\CRLBvanilla$). Further, the friendship paradox based MLEs have a negligible bias when the power-law exponent $\alpha$ is larger than $2.5$ and a slight increase in bias is observed when $\alpha$ decreases below~$2.5$ that could be due to the error of imposing a power-law degree distribution on a finite sized graph. The friendship paradox based MLE~$\fpMLE$ has a slightly smaller MSE compared to its discrete version~$\fpDiscreteMLE$ for smaller minimum degree values~($\mindegree\leq5$). Therefore, the numerical results verify the conclusion reached from the statistical analysis~(Sec.~\ref{sec:statistical_analysis}) that friendship paradox based MLEs~$\fpMLE, \fpDiscreteMLE$ yield a better accuracy compared to the vanilla methods based on uniform sampling~$\vanillaMLE, \vanillaDiscreteMLE$ and, the difference in the accuracy~(in terms of MSE) could be up to several orders of magnitude.

\section{Empirical Results using Real-World Network Datasets}
\label{sec:empirical_results}

In this section, we use several real-world network datasets to compare the performance of the four estimates presented in Sec.~\ref{sec:MLEs}: vanilla MLE~${\vanillaMLE}$ given in Eq.~\eqref{eq:vanilla_MLE},  friendship paradox based MLE~$\fpMLE$ given in Eq.~\eqref{eq:FP_MLE} and their discrete versions~$\vanillaDiscreteMLE, \fpDiscreteMLE$ given in Eq.~\eqref{eq:vanilla_discrete_MLE} and Eq.~\eqref{eq:FP_discrete_MLE} respectively. Our aim is to illustrate how the better performance of the proposed friendship paradox based methods~($\fpMLE, \fpDiscreteMLE$) compared to the widely used vanilla methods~($\vanillaDiscreteMLE, \fpDiscreteMLE$) seen from the statistical analysis~(Sec.~\ref{sec:statistical_analysis}) and numerical experiments~(Sec.~\ref{sec:numerical_results}) holds in real-world scenarios as well. 

\vspace{0.25cm}
\noindent
{\bf Description of Datasets: } We use 18 datasets obtained from the corpus used in the large scale study presented in~\cite{broido2019scale}. The datasets are publicly available also from the Colorado Index of Complex Networks~(ICON)~\cite{ICONColorado}. The 18 datasets include:
\begin{itemize}
	\item twelve biological networks~(metabolic networks of various species, protein interaction networks and a food-web network)
	
	\item three informational networks~(two networks of word adjacency in texts of two languages and the network of similarity scores among books sold on Amazon website)
	
	\item three technological networks~(three networks that represent dependencies in software systems  as complex networks).	
\end{itemize}
 The complementary cumulative distribution function~(complementary CDF) $\CCDF$ of the degree distribution (i.e.~$\CCDF(k)$ denotes the fraction of nodes with degrees greater than~$k$) of each network dataset was visually examined on double-logarithmic axes. Then, the minimum degree~$\mindegree$ is chosen such that the complementary CDF of the degree distribution~$\CCDF(k)$ is approximately a straight line on double-logarithmic axes in the region $k\geq\mindegree$. This ensures that the network datasets we use~(at least approximately) satisfy the power-law assumption~(Assumption~\ref{assumption:continuous_distribution}). Fig.~\ref{fig:EmpiricalResults_Biological}~(biological networks) and Fig.~\ref{fig:EmpiricalResults_Technological_and_Informational}~(informational and technological networks) show the chosen value of the minimum degree~$\mindegree$, the number of nodes with degrees greater than the minimum degree~$\ntail$~(which is effectively the number of nodes in the network for our purpose) and the complementary CDF~$\CCDF(k)$~(in blue color) for each considered dataset. The sizes of the datasets~$\ntail$ range from $226$ nodes to $268940$ nodes. Further information about the datasets are given in Appendix~\ref{append:reproducibility}.

 \vspace{0.25cm}
 \noindent
 {\bf Experimental Setup:} For each network dataset, the sample size~$\samplesize$ is set to be the minimum of $5\%$ of network size~$\ntail$~(rounded to the nearest integer) or $1000$. The four estimates $\vanillaMLE, \vanillaDiscreteMLE, \fpMLE, \fpDiscreteMLE$~(discussed in Sec.~\ref{sec:MLEs}) and their empirical variances are calculated using a Monte-Carlo average with $1000$ iterations. Further, the re-weighted Kolmogorov-Smirnov~(KS) statistic~(referred simply as the KS-statistic henceforth)~\cite{clauset2009power},
\begin{equation}
\label{eq:KS_statistic}
D(\hat{\alpha}) = \max_{k \geq \mindegree} \frac{|\CCDF_{\hat{\alpha}}(k) - \CCDF(k)|}{\sqrt{\CCDF(k)(1 - \CCDF(k))}}, \quad \hat{\alpha} \in \{\vanillaMLE, \vanillaDiscreteMLE, \fpMLE, \fpDiscreteMLE\}
\end{equation}
 is used to quantify the distance between the empirical complementary CDF $\CCDF$ and the complementary CDF~$\CCDF_{\hat{\alpha}}$ of a power-law distribution~(of the form Eq.~\eqref{eq:power_law_distribution}) with the estimated power-law exponent $\hat{\alpha} \in \{\vanillaMLE, \vanillaDiscreteMLE, \fpMLE, \fpDiscreteMLE\}$ via the four methods discussed in Sec.~\ref{sec:MLEs}. 
 
 \begin{figure}
 	\includegraphics[width=\columnwidth,  trim={0.19cm 0.23cm 0.1cm 0.1cm},clip]{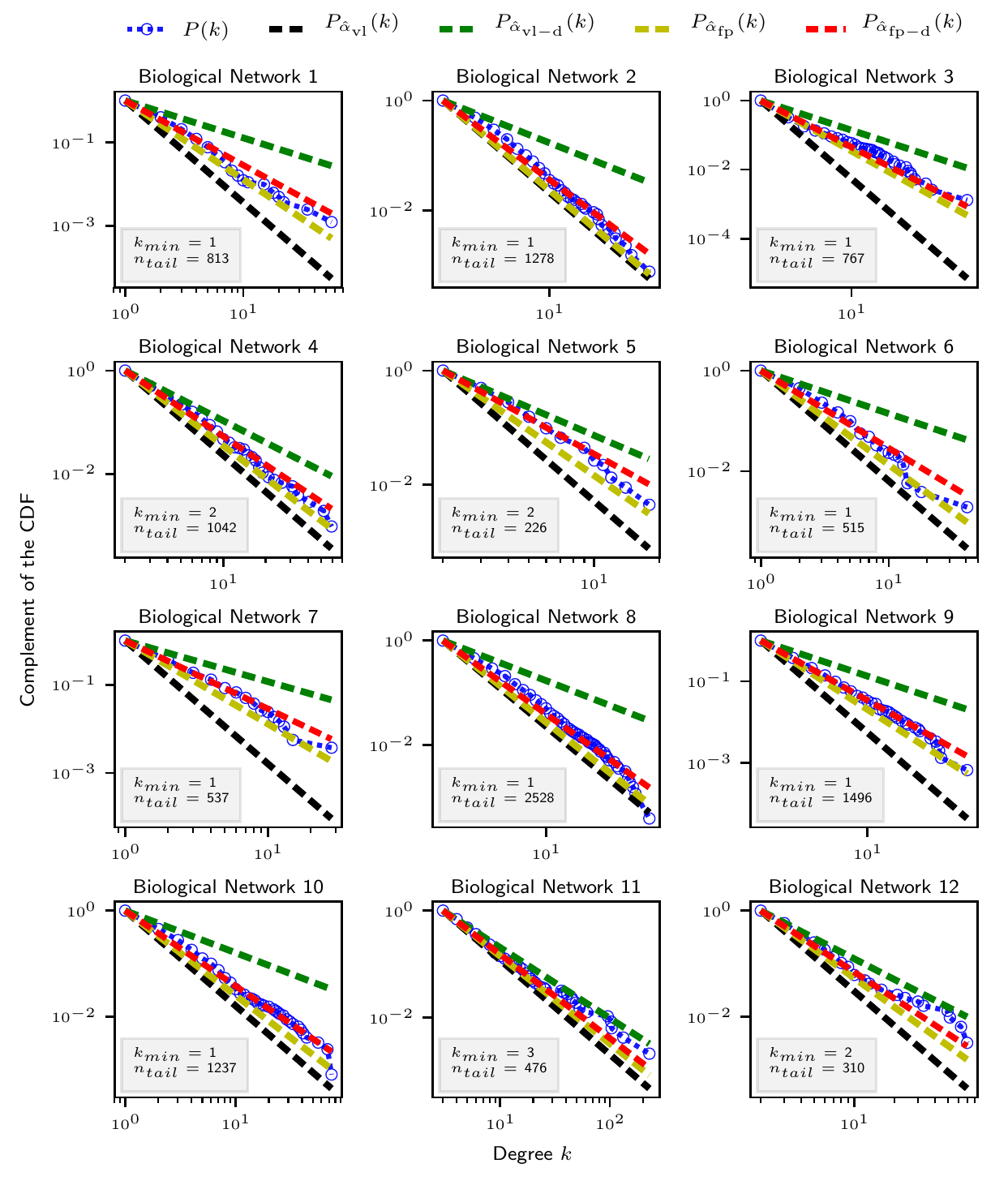}
 	\caption{
 		Plots show the complementary CDF~$\CCDF$ of the degree distribution of each of the~$12$ real-world biological network datasets~(discussed in Sec.~\ref{sec:empirical_results} and Appendix~\ref{append:reproducibility}) and the complementary CDFs of power-law distributions ($\CCDF_{\vanillaMLE}, \CCDF_{\vanillaDiscreteMLE}, \CCDF_{\fpMLE}, \CCDF_{\fpDiscreteMLE}$) with the power-law exponent values estimated from the four methods discussed in Sec.~\ref{sec:MLEs}~($\vanillaMLE, \vanillaDiscreteMLE, \fpMLE, \fpDiscreteMLE$). The estimated values and their variances are given in Table~\ref{table:empirical_results}. The plots indicate that the friendship paradox based maximum likelihood estimate~(MLE)~$\fpMLE$ and its discrete version~$\fpDiscreteMLE$ obtain better power-law fits compared to the vanilla methods based on uniform sampling~$\vanillaMLE,\vanillaDiscreteMLE$~(that either overestimate or underestimate the power-law exponent). This observation confirms the conclusion reached from statistical analysis~(Sec.~\ref{sec:statistical_analysis}) and numerical experiments~(Sec.~\ref{sec:numerical_results}) that the proposed friendship paradox based MLEs~$\fpMLE, \fpDiscreteMLE$ outperform the widely used vanilla MLEs based on uniform sampling~$\vanillaMLE, \vanillaDiscreteMLE$ for estimating power-law degree distributions. 
 	}
 	\label{fig:EmpiricalResults_Biological}
 \end{figure}

 \begin{figure}
 	\includegraphics[width=\columnwidth,  trim={0.19cm 0.2cm 0.1cm 0.55cm},clip]{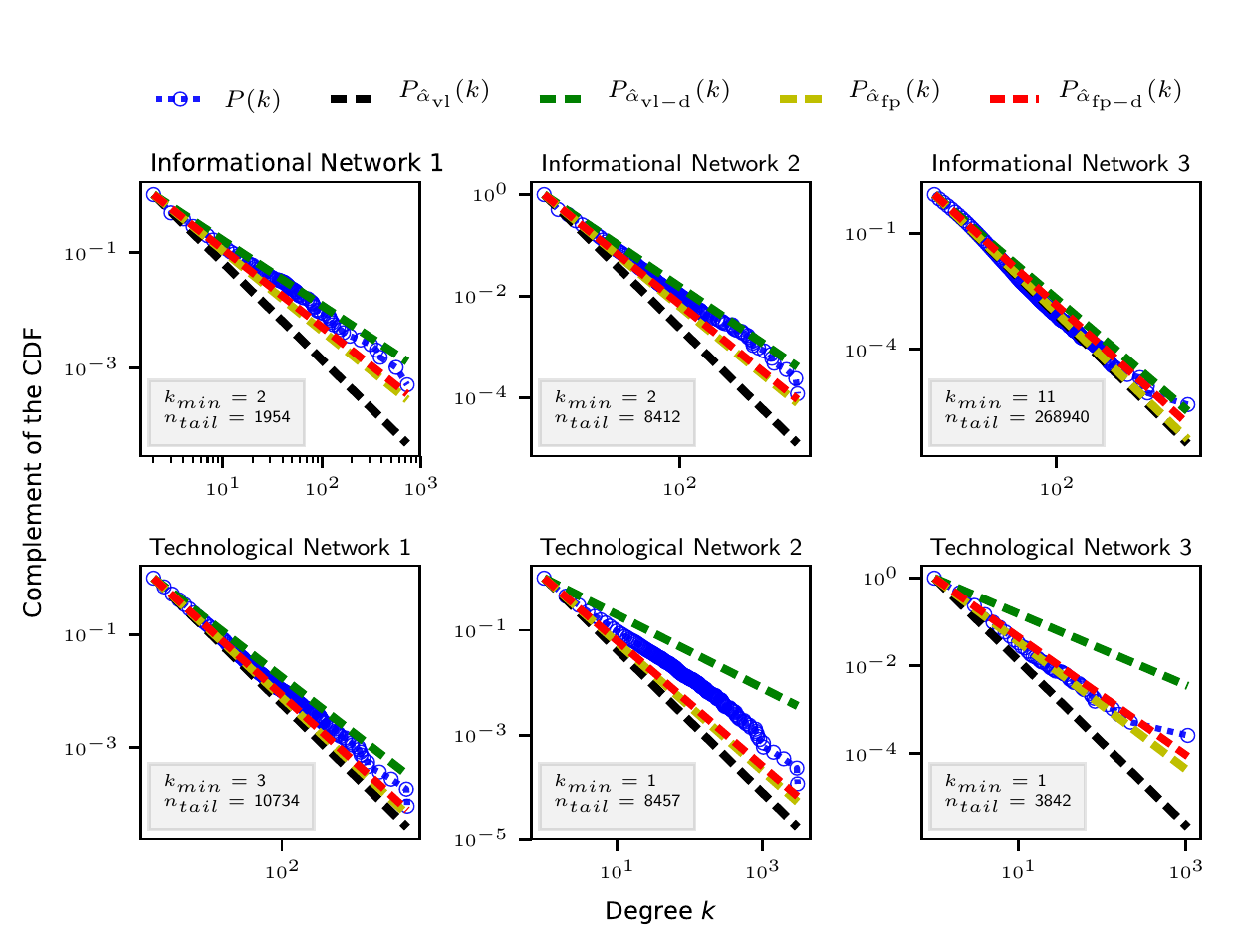}
 	\caption{
 		 		Plots show the complementary CDF~$\CCDF$ of the degree distribution of each of the~$6$ real-world information and technological network datasets~(discussed in Sec.~\ref{sec:empirical_results} and Appendix~\ref{append:reproducibility}) and the complementary CDFs of power-law distributions ($\CCDF_{\vanillaMLE}, \CCDF_{\vanillaDiscreteMLE}, \CCDF_{\fpMLE}, \CCDF_{\fpDiscreteMLE}$) with the power-law exponent values estimated from the four methods discussed in Sec.~\ref{sec:MLEs}~($\vanillaMLE, \vanillaDiscreteMLE, \fpMLE, \fpDiscreteMLE$). The estimated values and their variances are given in Table~\ref{table:empirical_results}. The plots indicate that the friendship paradox based maximum likelihood estimate~(MLE)~$\fpMLE$ and its discrete version~$\fpDiscreteMLE$ obtain better power-law fits compared to the vanilla MLEs based on uniform sampling~$\vanillaMLE,\vanillaDiscreteMLE$~(that either overestimate or underestimate the power-law exponent). This observation is similar to the one that is observed from biological networks~(Fig.~\ref{fig:EmpiricalResults_Biological}) and, also confirms the conclusion reached from statistical analysis~(Sec.~\ref{sec:statistical_analysis}) and numerical experiments~(Sec.~\ref{sec:numerical_results}) that the proposed friendship paradox based MLEs~$\fpMLE, \fpDiscreteMLE$ outperform the widely used vanilla MLEs based on uniform sampling~$\vanillaMLE, \vanillaDiscreteMLE$ for estimating power-law degree distributions. 
 	}
 	\label{fig:EmpiricalResults_Technological_and_Informational}
 \end{figure} 
 
 \vspace{0.1cm}
 \paragraph{Empirical Results}: Fig.~\ref{fig:EmpiricalResults_Biological}~(biological networks) and Fig.~\ref{fig:EmpiricalResults_Technological_and_Informational}~(informational and technological networks) illustrate the empirical complementary CDF~$\CCDF$~(using the $\ntail$ number of nodes in each dataset) and the complementary CDFs of the power-law distributions~(of the form Eq.~\eqref{eq:power_law_distribution}) with power-law exponent set to be the estimated values using the four method~$\vanillaMLE, \vanillaDiscreteMLE, \fpMLE, \fpDiscreteMLE$. Further, Table~\ref{table:empirical_results} summarizes the value of the KS-statistic~(given in Eq.~\eqref{eq:KS_statistic}) and the variance of each of the four estimates~$\vanillaMLE, \vanillaDiscreteMLE, \fpMLE, \fpDiscreteMLE$. Several important observations that complement the statistical analysis~(Sec.~\ref{sec:statistical_analysis}) and the numerical experiments~(Sec.~\ref{sec:numerical_results}) can be drawn from the empirical results~(Fig.~\ref{fig:EmpiricalResults_Biological}, Fig.~\ref{fig:EmpiricalResults_Technological_and_Informational} and Table~\ref{table:empirical_results}).
 
 First, it can be visually observed from Fig.~\ref{fig:EmpiricalResults_Biological} and Fig.~\ref{fig:EmpiricalResults_Technological_and_Informational} that the proposed friendship paradox based MLE~$\fpMLE$ and its discrete version~$\fpDiscreteMLE$ fit the empirical degree distributions better compared to the vanilla methods that are based on uniform sampling~$\vanillaMLE, \vanillaDiscreteMLE$. Especially, it can be seen that the vanilla MLE~$\vanillaMLE$ overestimates the value of the power-law exponent~$\alpha$ while its discrete version~$\vanillaDiscreteMLE$ underestimates it (compared to the two friendship paradox based methods), especially when the minimum degree is small~(i.e.~$\mindegree \leq 5$). This observation also agrees with the theoretical results~(Corollary~\ref{cor:bias_variance_ordering}) as well as the numerical experiments~(presented in Sec.~\ref{sec:numerical_results}) - for example, note from the sub-figure for $\mindegree = 1$ in Fig.~\ref{fig:Var_MSE_CRLB} that vanilla MLE~$\vanillaMLE$ has a positive bias and its approximation to discrete distributions~$\vanillaDiscreteMLE$ has a negative bias which are both larger in magnitude compared to the friendship paradox based MLEs. Apart from visual observation, the values of the KS-statistic given in Table~\ref{table:empirical_results} also verify that the friendship paradox based MLEs~$\fpMLE, \fpDiscreteMLE$ of the power-law exponent~$\alpha$ fit the empirical data better compared to the vanilla MLEs based on uniform sampling~$\vanillaMLE, \vanillaDiscreteMLE$. Of the two proposed friendship paradox based estimates, the discrete version~$\fpDiscreteMLE$ seems to have a better performance compared to~$\fpMLE$ in terms of the KS-statistic based comparison in all cases except two~(Informational Network~1 and Informational Network~2) - in those two cases the friendship paradox based MLE~$\fpMLE$ does better~(see Table~\ref{table:empirical_results}). Thus, the empirical results indicate that the friendship paradox based MLE~$\fpMLE$ and its discrete version~$\fpDiscreteMLE$ offer better accuracy in fitting power-law degree distributions to networks compared to the widely used vanilla MLEs based on uniform sampling~$\vanillaMLE, \vanillaDiscreteMLE$. 
 
 Second, the values of the variances of the four estimates~$\vanillaMLE, \vanillaDiscreteMLE, \fpMLE, \fpDiscreteMLE$ given in Table~\ref{table:empirical_results} indicate that the friendship paradox based methods have smaller variances compared to the vanilla methods based on uniform sampling. More specifically, the friendship paradox based discrete MLE~$\fpDiscreteMLE$ has the smallest variance amongst all four methods in all considered cases.  The difference in variance between $\fpDiscreteMLE$ and the vanilla methods is several orders of magnitude in most cases as seen from Table~\ref{table:empirical_results}. This observation also agrees with the conclusions reached from the statistical analysis~\ref{sec:statistical_analysis}~(Corollary~\ref{cor:bias_variance_ordering}) and the numerical experiments~(second column of Fig.~\ref{fig:Var_MSE_CRLB}). 
 
\vspace{0.25cm}
\noindent
{\bf Summary of Empirical Results: }This section used $18$ real-world network datasets to empirically evaluate the performance of the four methods presented in Sec.~\ref{sec:MLEs}: vanilla maximum likelihood estimate~${\vanillaMLE}$~(given in~Eq.~\eqref{eq:vanilla_MLE}),  friendship paradox based MLE~$\fpMLE$~(given in~Eq.~\eqref{eq:FP_MLE}) and their discrete versions~$\vanillaDiscreteMLE, \fpDiscreteMLE$~(given in~Eq.~\eqref{eq:vanilla_discrete_MLE} and~Eq.~\eqref{eq:FP_discrete_MLE}, respectively).  The empirical results indicate that the proposed friendship paradox based methods yield a better estimate of the power-law exponent for network data in the sense that the estimated power-law distribution results in a smaller~(in comparison to the vanilla methods)  Kolmogorov-Smirnov~(KS) statistic. Further, the friendship paradox based methods also result in a smaller~(up to several orders of magnitude) variance compared to the widely used vanilla methods that are based on uniform sampling. Of the two friendship paradox based methods, the friendship paradox based discrete MLE~$\fpDiscreteMLE$ shows a better accuracy compared to $\fpMLE$ in terms of both KS-statistic and variance in most cases. The empirical results support and verify the conclusion reached from the statistical analysis in Sec.~\ref{sec:statistical_analysis} and the numerical experiments in Sec.~\ref{sec:numerical_results}: the proposed friendship paradox based maximum likelihood estimation method yields better results in estimating the exponent of power-law degree distributions compared to the widely used vanilla methods based on uniform sampling.

\begin{landscape}
	\thispagestyle{mylandscape} 
	\centering

	\begin{table}
	\vspace{-0.1cm}
			\caption{Empirical results obtained by applying the estimates presented in Sec.~\ref{sec:MLEs} to $18$ real-world datasets}
	\centering
			\begin{tabular}{lll|llll|lll>{\bf}l|lll>{\bf}l}
			\toprule
			{} &  $n_{\mathrm{tai}l}$ & $\mindegree$ & $\vanillaMLE$ & $\vanillaDiscreteMLE$ & $\fpMLE$ & $\fpDiscreteMLE$ &      $D(\vanillaMLE)$ & $D(\vanillaDiscreteMLE)$ &       $D(\fpMLE)$ & $D(\fpDiscreteMLE)$ &       $\var\{\vanillaMLE\}$ &$\var\{\vanillaDiscreteMLE\}$ &        $\var\{\fpMLE\}$ & $\var\{\fpDiscreteMLE\}$\\
			Name                    &         &        &             &                      &             &                      &        &               &         &               &           &                &           &                \\
			\midrule
			Biological Network 1    &     813 &      1 &        3.44 &                 1.89 &        2.89 &                 2.55 &  0.549 &         0.383 &   0.286 &         0.116 &  4.36e-01 &       6.00e-03 &  1.73e-02 &       2.42e-03 \\
			Biological Network 2    &    1278 &      1 &        2.67 &                 1.77 &        2.63 &                 2.44 &  0.372 &         0.358 &   0.351 &         0.246 &  8.53e-02 &       3.68e-03 &  3.42e-03 &       7.91e-04 \\
			Biological Network 3    &     767 &      1 &        3.26 &                 1.86 &        2.46 &                 2.35 &  0.862 &         0.448 &   0.164 &         0.131 &  6.93e-01 &       1.13e-02 &  3.34e-03 &       1.03e-03 \\
			Biological Network 4    &    1042 &      2 &        3.34 &                 2.39 &        3.08 &                 2.82 &  0.321 &         0.193 &     0.2 &        0.0801 &  1.59e-01 &       1.93e-02 &  2.17e-02 &       6.83e-03 \\
			Biological Network 5    &     226 &      2 &        4.28 &                 2.63 &        3.63 &                 3.09 &  0.578 &          0.19 &    0.32 &         0.137 &  1.83e+00 &       9.37e-02 &  2.22e-01 &       3.94e-02 \\
			Biological Network 6    &     515 &      1 &        3.19 &                 1.85 &        2.87 &                 2.54 &  0.546 &         0.389 &   0.384 &         0.212 &  4.02e-01 &       9.04e-03 &  2.05e-02 &       2.79e-03 \\
			Biological Network 7    &     537 &      1 &        3.79 &                 1.93 &        2.89 &                 2.55 &  0.765 &         0.363 &   0.221 &        0.0788 &  8.89e-01 &       1.04e-02 &  3.26e-02 &       4.12e-03 \\
			Biological Network 8    &    2528 &      1 &        2.64 &                 1.77 &        2.56 &                  2.4 &  0.352 &         0.327 &   0.297 &         0.191 &  3.73e-02 &       1.66e-03 &  1.26e-03 &       3.38e-04 \\
			Biological Network 9    &    1496 &      1 &        3.26 &                 1.87 &        2.69 &                 2.46 &  0.471 &         0.353 &   0.152 &        0.0348 &  2.14e-01 &       4.07e-03 &  5.46e-03 &       1.12e-03 \\
			Biological Network 10   &    1237 &      1 &         2.8 &                 1.79 &         2.6 &                 2.42 &  0.402 &         0.351 &    0.28 &         0.169 &  8.94e-02 &       3.48e-03 &  3.49e-03 &       8.48e-04 \\
			Biological Network 11   &     476 &      3 &        2.77 &                 2.32 &        2.63 &                 2.57 &  0.187 &          0.17 &   0.141 &         0.117 &  1.86e-01 &       5.46e-02 &  1.21e-02 &       7.66e-03 \\
			Biological Network 12   &     310 &      2 &        3.17 &                  2.3 &        2.82 &                 2.66 &  0.395 &          0.19 &   0.205 &         0.139 &  6.45e-01 &       7.75e-02 &  4.27e-02 &       1.56e-02 \\
			Informational Network 1 &    1954 &      2 &        2.72 &                 2.14 &        2.39 &                 2.35 &  0.339 &           0.3 &   {\bf 0.174} &         {\normalfont 0.192} &  6.19e-02 &       1.21e-02 &  9.24e-04 &       5.99e-04 \\
			Informational Network 2 &    8412 &      2 &        2.56 &                 2.08 &        2.29 &                 2.27 &  0.245 &         0.286 &   {\bf 0.167} &         {\normalfont 0.179} &  1.16e-02 &       2.62e-03 &  1.02e-04 &       7.35e-05 \\
			Informational Network 3 &  268940 &     11 &        4.23 &                 3.81 &         4.2 &                 3.99 &  0.124 &        0.0699 &   0.115 &        0.0624 &  1.25e-02 &       7.13e-03 &  5.81e-03 &       3.93e-03 \\
			Technological Network 1 &   10734 &      3 &        2.46 &                 2.15 &        2.39 &                 2.36 &  0.111 &         0.102 &  0.0697 &        0.0547 &  4.29e-03 &       1.67e-03 &  2.05e-04 &       1.55e-04 \\
			Technological Network 2 &    8457 &      1 &        2.37 &                  1.7 &        2.23 &                 2.19 &  0.298 &         0.335 &   0.184 &         0.161 &  9.75e-03 &       6.54e-04 &  2.96e-05 &       1.66e-05 \\
			Technological Network 3 &    3842 &      1 &        2.87 &                 1.81 &        2.45 &                 2.34 &  0.418 &         0.388 &   0.198 &         0.137 &  3.29e-02 &       1.16e-03 &  9.16e-04 &       3.05e-04 \\
			\bottomrule
		\end{tabular}
	\label{table:empirical_results}
	\vspace{-0.4cm}
	\end{table}

		\justify
		\begin{compactitem}
			\item Table~\ref{table:empirical_results} shows:
			\begin{compactenum}[i.)]
				\item values of the four estimates presented in Sec.~\ref{sec:MLEs}: vanilla MLE~${\vanillaMLE}$~(Eq.~\eqref{eq:vanilla_MLE}),  friendship paradox based MLE~$\fpMLE$~(Eq.~\eqref{eq:FP_MLE}) and their discrete versions~$\vanillaDiscreteMLE$~(Eq.~\eqref{eq:vanilla_discrete_MLE})  and $\fpDiscreteMLE$~(Eq.~\eqref{eq:FP_discrete_MLE}), obtained using~$18$ real-world datasets discussed in Sec.~\ref{sec:empirical_results}. 
				
				\item Kolmogorov-Smirnov~(KS) statistic $D(\hat{\alpha})$ given in Eq.~\eqref{eq:KS_statistic} which quantifies the distance between the empirical complementary CDF $\CCDF$ and the complementary CDF~$\CCDF_{\hat{\alpha}}$ of a power-law distribution~(of the form Eq.~\eqref{eq:power_law_distribution}) with the estimated power-law exponent $\hat{\alpha} \in \{\vanillaMLE, \vanillaDiscreteMLE, \fpMLE, \fpDiscreteMLE\}$. The values in bold-font indicate the smallest value of the KS-statistic for each dataset.
				
				\item Variance of each of the four estimates $\var\{\vanillaMLE\}, \var\{\vanillaDiscreteMLE\}, \var\{\fpMLE\}, \var\{\fpDiscreteMLE\}$ with the smallest of them for each dataset highlighted in bold-font.
			\end{compactenum}
			\item The values shown in the Table~\ref{table:empirical_results} indicate that the friendship paradox based MLEs, especially the discrete version~$\fpDiscreteMLE$, yield the best results amongst all four considered methods. This observation agrees with both statistical analysis~(Sec.~\ref{sec:statistical_analysis}) and numerical experiments~(Sec.~\ref{sec:numerical_results}).
		\end{compactitem}

\end{landscape}

\section{Conclusion and Extensions}
\noindent
{\bf Conclusion: }This paper proposed a maximum likelihood estimator that exploits the \textit{friendship paradox} based sampling of nodes in order to estimate a power-law degree distribution of a network. Most widely used parametric methods such as linear regression, vanilla maximum likelihood estimation with uniformly sampled nodes can produce unreliable results due to the lack of samples from the characteristic heavy tail of the degree distribution. Other non-parametric methods have used friendship paradox based sampling (to yield more samples from the heavy tail of the degree distribution) but do not exploit the explicit parametric form of the power-law degree distribution and also do not provide any analytical quantification of the estimation accuracy. In contrast, by exploiting both the friendship paradox based sampling (to capture the heavy tail of the degree distribution) and the explicit parametric form of the power-law degree distribution (to derive the precise expression of the estimator), the proposed method results in a provably better accuracy compared to the state of the art methods. 
The analytical results show that the proposed method  is strongly consistent and has a smaller bias, a smaller variance and a smaller Cram\`{e}r-Rao lower bound compared to the widely used vanilla maximum likelihood estimation method~(with uniform samples). Numerical and empirical results illustrate the better accuracy of the proposed method~(compared to widely used alternative methods) including its properties such as smaller bias, near optimal variance and sample efficiency. All analytical, numerical and empirical results lead to the conclusion that the proposed friendship paradox based maximum likelihood estimation method achieves an improved accuracy compared to the widely used vanilla maximum likelihood estimation method in both finite and asymptotic sample size regimes. It was also shown that the proposed method easily generalizes to degree distributions with parametric forms other than power-law~(such as exponential degree distributions) while preserving all its desirable statistical properties. 

\vspace{0.25cm}
\noindent
{\bf Limitations and Extensions:} There are several limitations of the proposed method that open up future research directions. Firstly, the proposed estimate is based on the assumption that the underlying social network has a power-law degree distribution. However, recent works (e.g.~\cite{broido2019scale}) have pointed out that there might be other types of degree distributions that fit some real-world networks better than a power-law. Though our results indicate that the proposed method easily generalizes to other types of degree distributions such as the exponential degree distribution~(Appendix~\ref{append:MLE_exponential}), it remains to be extended to other types of distributions such as the log-normal distribution which is also prevalent in real-world networks. Secondly, extending the proposed framework to directed networks by using the versions of friendship paradox for directed graphs~\cite{alipourfard2019friendship,higham2018centrality} is of importance due to the power-law behavior of the in- and out- degree distributions of real world directed networks~\cite{bollobas2003directed, cohen2003directed}. 
Such extensions might also be able to exploit the correlations~\cite{williams2014degree} between in- and out- degrees to make the estimation further efficient. Thirdly, extending the proposed framework to a setting where the aim is to track a time-evolving power-law degree distribution~(i.e.~the power-law exponent~$\alpha$ of the degree distribution evolves over time) is also an interesting future research direction. Bayesian and stochastic approximation methods~(such as \cite{hamdi2014_tracking, krishnamurthy2017tracking}) 
based on friendship paradox based sampling might be suitable for such contexts. Finally, the method we proposed in this paper is based on the original version of the friendship paradox~(Theorem~\ref{th:friendship_paradox_Feld}). Exploring how other versions of the friendship paradox~(such as spectral friendship paradox~\cite{higham2018centrality}) can be used in estimating network properties other than degree distributions~(e.g.~spectra of networks) also remains an interesting future direction.

\begin{acks}
	The authors thank Aaron Clauset at University of Colorado at Boulder and Johan Ugander at Stanford University for helpful comments.  This research was supported by the U.~S. Army Research Office under grant W911NF-19-1-0365, and the National Science Foundation under grant 1714180.
\end{acks}

\newpage 
\appendix

\section{Proof of Theorem~\ref{th:bias_variance_fpMLE}}
\label{append:bias_variance_fpMLE}

	The proof of Theorem~\ref{th:bias_variance_fpMLE} is inspired by the ideas used in the proof of Theorem~\ref{th:bias_variance_vanillaMLE}. Note from the expression for $\fpMLE$ in Eq.~\eqref{eq:FP_MLE} that, 
	\begin{align}
		\label{eq:bias_variance_fpMLE_appendix}
		\begin{split}
			\bias\{\fpMLE\} &= \mathbb{E}\{\fpMLE\} - \alpha 
			= \samplesize\mathbb{E}\Bigg\{		\frac{1}{\sum_{i = 1}^{\samplesize}\ln\Big(\frac{d(Y_i)}{\mindegree}\Big)} \Bigg\} + 2 - \alpha \; \text{,}\\
			\var\{\fpMLE\} & = \samplesize^2\var \Bigg\{		\frac{1}{\sum_{i = 1}^{\samplesize}\ln\Big(\frac{d(Y_i)}{\mindegree}\Big)}	\Bigg\}.
		\end{split}
	\end{align}
	It can easily be shown that 
	\begin{equation}
	\label{eq:expoential_cdf_of_log_dY}
	{\mathbb{P}\Big\{\ln\Big(\frac{d(Y_i)}{\mindegree}\Big) \leq c\Big\} = 1 -   e^{-c(\alpha - 2)}}\quad i = 1,\dots, \samplesize, \nonumber
	\end{equation}
	and therefore, $\ln\Big(\frac{d(Y_i)}{\mindegree}\Big) \; i = 1,\dots, \samplesize$ are independent and identically distributed~(iid) exponential random variables with parameter $\alpha - 2$. Since the sum of iid exponential random variables is a Gamma random variable, it follows that
	\begin{align}
		\frac{1}{\sum_{i = 1}^{\samplesize}\ln\Big(\frac{d(Y_i)}{\mindegree}\Big)} &\sim \text{Inv-Gamma}(\samplesize, \alpha - 2) 	\label{eq:inverse_Gamma_dist}
	\end{align}
	where, ${\text{Inv-Gamma}(\samplesize, \alpha - 2)}$ denotes Inverse-Gamma distribution~(distribution of the reciprocal of a Gamma distributed random variable) with parameters $\samplesize$ and $\alpha -2$. Eq.~\eqref{eq:inverse_Gamma_dist} then implies that~\cite{witkovsky2001computing},
	\begin{align}
		\label{eq:mean_variance_inv_gamma_appendix}
		\begin{split}
			\mathbb{E}\Bigg\{		\frac{1}{\sum_{i = 1}^{\samplesize}\ln\Big(\frac{d(Y_i)}{\mindegree}\Big)}\Bigg\} &= \frac{\alpha -2}{\samplesize-1},\; \textbf{ \normalfont for } \samplesize > 1\\
			\var\Bigg\{		\frac{1}{\sum_{i = 1}^{\samplesize}\ln\Big(\frac{d(Y_i)}{\mindegree}\Big)}\Bigg\}	&= \frac{(\alpha -2)^2}{(\samplesize - 1)^2(\samplesize - 2)},\; \textbf{ \normalfont for } \samplesize > 2.
		\end{split}
	\end{align}
	The result follows by substituting Eq.~\eqref{eq:mean_variance_inv_gamma_appendix} in the expressions for bias and variance in Eq.~\eqref{eq:bias_variance_fpMLE_appendix}. 

\section{Proof of Theorem~\ref{th:strong_consistency}}
\label{append:strong_consistency}
	The asymptotic unbiasedness and strong consistency of the vanilla MLE~$\vanillaMLE$ have been established in~\cite{muniruzzaman1957measures}.  We establish the strong consistency of the proposed friendship paradox based MLE~$\fpMLE$ below. 

Note from Eq.~\eqref{eq:bias_variance_vanillaMLE} and Eq.~\eqref{eq:bias_variance_FPmle} that bias of both vanilla MLE~$\vanillaMLE$ and friendship paradox based MLE~$\fpMLE$ goes to zero as the sample size $\samplesize$ tends to infinity. Hence, both estimates are asymptotically unbiased.  

Next, note from Eq.~\eqref{eq:FP_MLE} and Eq.~\eqref{eq:expoential_cdf_of_log_dY} that,
\begin{equation}
\fpMLE = \frac{1}{\bar{G}_\samplesize} + 2
\end{equation}
where,~$\bar{G}_n$ is the empirical mean of $\samplesize$ iid exponential random variables with parameter~$\alpha - 2$. Hence, by strong law of large numbers,~$\bar{G}_\samplesize$ converges almost surely to~$1/(\alpha - 2)$. Since the friendship paradox based MLE~$\fpMLE$ is a continuous function of~$\bar{G}_n$, $\fpMLE$ converges almost surely to~$\alpha$. 

\section{Proof of Theorem~\ref{th:CRLB}}
\label{append:CRLB}

	By definition~\cite{casella2002statistical}
\begin{align}
	\label{eq:definition_CRLB_FP_appendix}
	\CRLBfp(\alpha) &= -\Bigg[{\mathbb{E}\bigg\{\frac{\partial^2 \loglikelihoodY}{\partial^2 \alpha} \bigg\}}\Bigg]^{-1}
\end{align}
where, $\loglikelihoodY$ is the log-likelihood function for the friendship paradox based sampling defined in Eq.~\eqref{eq:loglikelihood_Y}. By evaluating Eq.~\eqref{eq:definition_CRLB_FP_appendix}, we get,
\begin{align}
	\CRLBfp(\alpha) &= -\Bigg[{\mathbb{E}\bigg\{ -\frac{\samplesize}{(\alpha -2)^2} \bigg\}}\Bigg]^{-1} = \frac{(\alpha -2)^2}{\samplesize}. \nonumber
\end{align}
Then, the fact that $q(k) \propto kp(k)$ (as discussed in Remark~\ref{remark:powerlaw_form_of_q}) can be used to obtain the expression for $\CRLBvanilla(\alpha)$ as a special case of $\CRLBfp(\alpha)$ by replacing $\alpha$ with $\alpha + 1$.

\section{Extension of results to exponential degree distributions}
\label{append:MLE_exponential}
In this appendix, we briefly explain how the main results of this paper 
extend to exponential degree distributions.

Analogous to the power-law case, the following assumption is made for the derivation and analysis of MLEs for exponential case.
\begin{assumption}
	The exponential distribution is continuous in $k$ and is of the form,
	\begin{equation}
	\label{eq:exponential_distribution}
	\degdist_{exp}(k) = \lambda e^{-\lambda k}, \quad k\geq 0.
	\end{equation}
\end{assumption}

\subsection{Maximum Likelihood Estimation of the Exponential Distribution}

\noindent
{\bf Vanilla MLE for exponential degree distribution: } The vanilla maximum likelihood estimation of the rate parameter $\lambda$ of the degree distribution~(defined in Eq.~\eqref{eq:exponential_distribution}) begins by sampling~$n$~nodes $X_1, X_2, \dots, X_{\samplesize}$ independently and uniformly from the network. Then, the likelihood of observing the degree sequence $d(X_1), 
\dots, d(X_{\samplesize})$~is,
\begin{align}
\mathbb{P}\{d(X_1), \dots, d(X_{\samplesize})\vert \lambda \} = \prod_{i = 1}^{\samplesize} \lambda e^{-\lambda d(X_i)} \nonumber
\end{align}
following~Eq.~\eqref{eq:exponential_distribution}. Therefore, the log-likelihood for vanilla MLE of the rate parameter $\lambda$ is,
\begin{align}
\loglikelihoodX(\lambda) = \ln \mathbb{P}\{d(X_1), 
\dots, d(X_{\samplesize})\vert \lambda \} 
=\samplesize\ln (\lambda)  -\lambda \sum_{i = 1}^{\samplesize} d(X_i). \nonumber
\end{align}
Then, by solving $\frac{\partial \loglikelihoodX}{\partial \lambda} = 0$, we get the vanilla MLE of the rate parameter $\lambda$ as,
\begin{align}
\label{eq:vanilla_MLE_exponential}
\vanillaMLEexponential = \frac{\samplesize}{\sum_{i = 1}^{\samplesize}d(X_i)}.
\end{align}

Next, we present a maximum likelihood estimator for the exponential distribution that exploits the friendship paradox.

\vspace{0.1cm}
\noindent
{\bf MLE for exponential degree distribution with friendship paradox based sampling: } 
Recall (from Eq.~\eqref{eq:qk_proportional_kpk}) that the neighbor degree distribution is defined by ${\neighbordegdist_{exp}(k) \propto k\degdist_{exp}(k)}$ and hence, its normalizing constant is given by the mean $\lambda^{-1}$ of the degree distribution~$\degdist_{exp}$. Therefore, the neighbor degree distribution for a network with an exponential distribution (of the form Eq.~\eqref{eq:exponential_distribution}) is given by,
\begin{align}
\label{eq:exponential_neighbor_degree_distribution}
\neighbordegdist_{exp}(k) = \frac{k\degdist_{exp}(k)}{\lambda^{-1}}
= \lambda^2ke^{-\lambda k} 
= \text{Gamma}(2,\lambda).
\end{align}
The friendship paradox based maximum likelihood estimator for the rate parameter $\lambda$ begins with sampling $\samplesize$ number of random neighbors $Y_1, Y_2,\dots, Y_{\samplesize}$ from the network independently. Then, the likelihood of observing the neighbor degree sequence $d(Y_1), d(Y_2), \dots, d(Y_{\samplesize})$ can be written using the neighbor degree distribution in Eq.~\eqref{eq:exponential_neighbor_degree_distribution} as, 
\begin{equation}
\hspace{-0.0cm}\mathbb{P}\{d(Y_1), \dots, d(Y_{\samplesize})\vert \alpha \} = \prod_{i = 1}^{\samplesize} \lambda^2d(Y_i)e^{-\lambda d(Y_i)}. \nonumber
\end{equation}
Therefore, the log-likelihood for the rate parameter $\lambda$ is,
\begin{align}
\loglikelihoodY(\lambda) = \ln \mathbb{P}\{d(Y_1), \dots, d(Y_{\samplesize})\vert \lambda \}
\hspace{-0.0cm} =2\samplesize\ln (\lambda) + \sum_{i=1}^{\samplesize} \Big(\ln(d(Y_i)) - \lambda d(Y_i)\Big) \nonumber
\end{align}
Then, by solving $\frac{\partial \loglikelihoodY}{\partial \lambda} = 0$, we get the FP based MLE of the rate parameter~$\lambda$ as,
\begin{align}
\label{eq:FP_MLE_exponential}
\fpMLEexponential = \frac{2\samplesize}{\sum_{i = 1}^{\samplesize}d(Y_i)}.
\end{align}

\subsection{Statistical Analysis of MLEs for Exponential Degree Distribution}

This section presents the statistical analysis (analogous to the results in Sec.~\ref{sec:statistical_analysis}) of the two MLEs for the exponential degree distribution: vanilla MLE~$\vanillaMLEexponential$ in Eq.~\eqref{eq:vanilla_MLE_exponential} and friendship paradox based MLE~$\fpMLE$ in Eq.~\eqref{eq:FP_MLE_exponential}. 

\paragraph{Comparison of bias and variance for finite sample size: } The following two results characterizes the bias and variance of the vanilla MLE~$\vanillaMLEexponential$ and the proposed friendship paradox based MLE~$\fpMLEexponential$ under a finite samples size~$\samplesize < \infty$.

\begin{theorem}
	\label{th:bias_variance_vanillaMLE_exponential}
	The bias and variance of the vanilla MLE~$\vanillaMLEexponential$ in Eq.~\eqref{eq:vanilla_MLE_exponential} for sample size~$\samplesize$ are given by,
	\begin{align}
	\begin{split}
	\bias\{\vanillaMLEexponential\} &= \frac{\lambda}{\samplesize -1 },\;\; \textbf{ \normalfont for } \samplesize > 1\\  \var \{\vanillaMLEexponential\} &= \frac{\samplesize^2\lambda^2}{(\samplesize-1)^2(\samplesize-2)},\;\; \textbf{ \normalfont for } \samplesize > 2.
	\end{split}
	\label{eq:bias_variance_vanillaMLE_exponential}
	\end{align}
\end{theorem}
\begin{proof}
	Note that $\sum_{i=1}^{\samplesize}d(X_i)$ (in Eq.~\eqref{eq:vanilla_MLE_exponential}) is a sum of iid random variables sampled from the exponential distribution specified in~Eq.~\eqref{eq:exponential_distribution}. Since sum of iid exponential random variables is a Gamma random variable, it follows that,
	\begin{equation}
	\frac{1}{\sum_{i = 1}^{\samplesize}d(X_i)} \sim \text{Inv-Gamma}(\samplesize, \lambda)
	\end{equation}
	and hence,
	\begin{align}
	\begin{split}
	\label{eq:mean_variance_inv_gamma_appendix_vanilla_exponential}
	\mathbb{E}\Bigg\{		\frac{1}{\sum_{i = 1}^{\samplesize} d(X_i)}\Bigg\} &= \frac{\lambda}{\samplesize-1},\;\; \textbf{ \normalfont for } \samplesize > 1\\
	\var\Bigg\{		\frac{1}{\sum_{i = 1}^{\samplesize} d(X_i)}\Bigg\} 	&= \frac{\lambda^2}{(\samplesize - 1)^2(\samplesize - 2)},\;\; \samplesize > 2.
	\end{split}
	\end{align}
	Then, the result follows from Eq.~\eqref{eq:vanilla_MLE_exponential} and Eq.~\eqref{eq:mean_variance_inv_gamma_appendix_vanilla_exponential}.
\end{proof}

\begin{theorem}
	\label{th:bias_variance_fpMLE_exponential}
	The bias and variance of the friendship paradox based MLE MLE~$\fpMLEexponential$ in Eq.~\eqref{eq:FP_MLE_exponential} for sample size~$\samplesize$ are given by,
	\begin{align}
	\begin{split}
	\bias\{\fpMLEexponential\} &= \frac{\lambda}{2\samplesize -1}, \;\; \textbf{ \normalfont for } \samplesize > 1\\ \var \{\fpMLEexponential\} &= \frac{2\samplesize^2\lambda^2}{(2\samplesize-1)^2(\samplesize-1)},\;\; \textbf{ \normalfont for } \samplesize > 2. 
	\label{eq:bias_variance_FPmle_exponential}
	\end{split}	
	\end{align}
\end{theorem}
\begin{proof}
	Note that $\sum_{i=1}^{\samplesize}d(Y_i)$~(in Eq.~\eqref{eq:FP_MLE_exponential}) is sum of iid random variables sampled from the Gamma distribution in Eq.~\eqref{eq:exponential_neighbor_degree_distribution}. Since sum of iid Gamma random variables is a Gamma random variable, it follows that,
	\begin{equation}
	\frac{1}{\sum_{i = 1}^{\samplesize}d(Y_i)} \sim \text{Inv-Gamma}(2\samplesize, \lambda).
	\end{equation}
	and hence,
	\begin{align}
	\begin{split}
	\label{eq:mean_variance_inv_gamma_appendix_fp_exponential}
	\mathbb{E}\Bigg\{		\frac{1}{\sum_{i = 1}^{\samplesize} d(Y_i)}\Bigg\} &= \frac{\lambda}{2\samplesize-1},\;\; \textbf{ \normalfont for } \samplesize > 1\\
	\var\Bigg\{		\frac{1}{\sum_{i = 1}^{\samplesize} d(Y_i)}\Bigg\} 	&= \frac{\lambda^2}{(2\samplesize - 1)^2(2\samplesize - 2)},\;\; \textbf{ \normalfont for } \samplesize > 2.\\
	\end{split}
	\end{align}
	Then, the result follows from Eq.~\eqref{eq:FP_MLE_exponential} and Eq.~\eqref{eq:mean_variance_inv_gamma_appendix_fp_exponential}. 
\end{proof}

The following consequence of Theorem~\ref{th:bias_variance_vanillaMLE_exponential} and Theorem~\ref{th:bias_variance_fpMLE_exponential} shows that the proposed friendship paradox based MLE~$\fpMLEexponential$ outperforms vanilla MLE~$\vanillaMLEexponential$ for any finite sample size~$\samplesize$ for exponential degree distributions.

\begin{corollary}
	\label{cor:bias_variance_ordering_exponential}
	The bias and variance of the vanilla MLE~$\vanillaMLEexponential$ defined in Eq.~\eqref{eq:vanilla_MLE_exponential} and the friendship paradox based MLE~$\fpMLEexponential$ defined in Eq.~\eqref{eq:FP_MLE_exponential} satisfy,
	\begin{align}
	\begin{split}
	\bias\{\fpMLEexponential\} &< \bias\{\vanillaMLEexponential\},\;\; \textbf{ \normalfont for } \samplesize > 1\\
	\var\{\fpMLEexponential\} &< \var\{\vanillaMLEexponential\},\;\; \textbf{ \normalfont for } \samplesize > 2.
	\end{split}
	\end{align}
\end{corollary}
We now turn to the case where the sample size~$\samplesize$ tends to infinity. 

\paragraph{Comparison of the asymptotic properties of the MLEs: }
We have the following two results (whose proofs follow from arguments similar to Theorem~\ref{th:strong_consistency} and Theorem~\ref{th:CRLB})
for the MLEs of rate parameter~$\lambda$ of the exponential degree distribution:
\begin{theorem}
	\label{th:strong_consistency_exponential}
	The vanilla MLE~$\vanillaMLEexponential$~(defined in Eq.~\eqref{eq:vanilla_MLE_exponential} and friendship paradox based MLE~$\fpMLEexponential$~(defined in Eq.~\eqref{eq:FP_MLE_exponential}) are asymptotically unbiased and strongly consistent~i.e.~converges to the true rate parameter~$\lambda$ with probability~$1$ as the sample size $\samplesize$ tends to infinity.  
\end{theorem} 


\begin{theorem}
	\label{th:CRLB_exponential}
	The  Cram\`{e}r-Rao Lower Bounds of the vanilla MLE $\vanillaMLEexponential$ in Eq.~\eqref{eq:vanilla_MLE_exponential} and the friendship paradox based MLE~$\fpMLEexponential$ in Eq.~\eqref{eq:FP_MLE_exponential} are given by,
	\begin{align}
	\hspace{-0.25cm}\CRLBvanilla(\lambda) = \frac{\lambda^2}{\samplesize}, \quad	\CRLBfp(\lambda) = \frac{\lambda^2}{2\samplesize}.
	\end{align}
\end{theorem}

\vspace{0.1cm}
\noindent
{\bf Summary of Results for Exponential Degree Distributions: } The analytical results presented in Sec.~\ref{sec:MLEs} and Sec.~\ref{sec:statistical_analysis} for power-law degree distributions extend to exponential degree distributions as well. More specifically, the proposed friendship paradox based MLE~$\fpMLEexponential$~(defined in Eq.~\eqref{eq:FP_MLE_exponential}) outperforms the vanilla MLE~$\vanillaMLEexponential$~(defined in Eq.~\eqref{eq:vanilla_MLE_exponential}) in the finite sample regime (in terms of bias and variance) as well as the asymptotic regime~(in terms of CRLB).

\section{Details for Reproducing the Numerical and Empirical Results}
\label{append:reproducibility}

This appendix provides additional details on the simulations that produced all numerical results~(Fig.~\ref{fig:Var_MSE_CRLB} and Fig.~\ref{fig:SampleSize_vs_MSE} in Sec.~\ref{sec:numerical_results}) and empirical results~(Fig.~\ref{fig:EmpiricalResults_Biological}, Fig.~\ref{fig:EmpiricalResults_Technological_and_Informational} and Table~\ref{table:empirical_results} in Sec.~\ref{sec:empirical_results}) 
in order to facilitate the full reproducibility of the results in this paper.

\paragraph{\bf Code Availability:} {Codes that produced all three figures in Sec.~\ref{sec:numerical_results} and Sec.~\ref{sec:empirical_results} are available in the Github repository~\cite{anonymousGitHubRepo2020KDD}. Further, a detailed outline of the steps followed in each numerical experiment as well as empirical methodology with addition details is presented below.}

\paragraph{\bf Reproducibility:}
All simulations were performed using version 3.7.3 of the software \textit{Python}. The random number generators were initialized with seed 123 (using the \textit{random.seed} function) in each simulation. 

\subsection{Simulation Details for Numerical Experiments}
\paragraph{\bf Simulation Details for Fig.~\ref{fig:Var_MSE_CRLB}}
{Twenty five} equally spaced values of the parameter $\alpha$ in the range from $2.1$ to $3.5$~(satisfying the Assumption~\ref{assumption:alpha_range} and including typical values observed in most real world networks)  were considered. For each considered value of $\alpha$ and $\mindegree$,  $50000$ random variables from the power-law distribution in Eq.~\eqref{eq:power_law_distribution} were sampled (using the \textit{scipy.stats.pareto.rvs} function) independently and rounded off to the nearest integer. 
The degree sequences produced using this simulation setup are guaranteed to satisfy Assumptions~\ref{assumption:continuous_distribution},~\ref{assumption:alpha_range}. This helps to rule out the possibility that the conclusions drawn from the numerical results could be due to violations of assumptions made on the underlying parametric model (such as not possessing a power-law degree distribution or parameters out of the considered range). Then, $n = 100$ elements (denoted by $d(X_1), \dots, d(X_\samplesize)$) were drawn from the degree sequence by uniformly and independently sampling. Another $n = 100$ elements (denoted by $d(Y_1), \dots, d(Y_\samplesize)$) were drawn from the degree sequence by independently sampling with a probability proportional to the value of each element.\footnote{This is equivalent to sampling an edge uniformly and then taking the degree of the node at a random end of that edge~i.e.~equivalent to sampling from the neighbor degree distribution~$\neighbordegdist$ given in Eq.~\eqref{eq:qk_proportional_kpk}. This is because the neighbor degree distribution~$\neighbordegdist(k)$ is the same as degree distribution~$\degdist(k)$ scaled by degree~$k$. } Next, the vanilla MLEs~$\vanillaMLE, \vanillaDiscreteMLE$ (by plugging in $d(X_1), \dots, d(X_\samplesize)$ to expressions in Eq.~\eqref{eq:vanilla_MLE} and Eq.~\eqref{eq:vanilla_discrete_MLE}, respectively)  and friendship paradox based MLEs~$\fpMLE, \fpDiscreteMLE$ (by plugging in $d(Y_1), \dots, d(Y_\samplesize)$ to expression in Eq.~\eqref{eq:FP_MLE} and Eq.~\eqref{eq:FP_discrete_MLE}, respectively) were computed and stored. This process was repeated  {$5000$} times (for each considered value of the power-law exponent~$\alpha$ and the minimum degree~$\mindegree$) and the empirically estimated bias, variance and MSE using all  {$5000$} independent iterations were plotted in Fig.~\ref{fig:Var_MSE_CRLB}. The Cram\`{e}r-Rao Lower Bounds were computed by directly plugging in the values of the power-law exponent~$\alpha$ into the expressions in Theorem~\ref{th:CRLB}.

\paragraph{\bf Simulation Details for Fig.~\ref{fig:SampleSize_vs_MSE}}
For Fig.~\ref{fig:SampleSize_vs_MSE}, the steps followed for obtaining Fig.~\ref{fig:Var_MSE_CRLB} (detailed above) were repeated for four different values of the sample size $n = 50, 100, 150, 200$ for each considered value of the minimum degree~$\mindegree$. The obtained empirical MSE values for different values of the sample size~$\samplesize$ are plotted in different colors in Fig.~\ref{fig:SampleSize_vs_MSE}.

\subsection{Details on the Real-World Datasets}
The $18$ network degree sequences that we use in Sec.~\ref{sec:empirical_results} are obtained from the dataset used for the study presented in~\cite{broido2019scale}. The full network datasets are also available via the ICON network data library~\cite{ICONColorado}. The degree sequence files that we use~(which were extracted by the authors of \cite{broido2019scale}) for each of the $18$ datasets and the code that produces the empirical results in Sec.~\ref{sec:empirical_results} are available in the Github repository~\cite{anonymousGitHubRepo2020KDD}. 

%
%
%
%
%
%
%
%
%
%

\end{document}